\documentclass[journal]{IEEEtran}
\usepackage[dvipsnames]{xcolor}
\usepackage[cmex10]{amsmath}
\usepackage{graphicx,times,color,amssymb,cite}
\usepackage{epstopdf}
\usepackage[disable]{todonotes} 
\usepackage{amsthm}
\usepackage{rotating}
\usepackage{fancyvrb}
\usepackage{epstopdf}
\newtheorem{theorem}{Theorem}
\newtheorem{assumption}{Assumption}
\newtheorem{lemma}{Lemma}
\newtheorem{remark}{Remark}
\newtheorem{definition}{Definition}

\makeatletter
\g@addto@macro\normalsize{%
	\setlength\abovedisplayskip{2pt}
	\setlength\belowdisplayskip{2pt}
	\setlength\abovedisplayshortskip{2pt}
	\setlength\belowdisplayshortskip{2pt}
}
\usepackage[font=small,skip=2pt]{caption}

\def\ee{{\epsilon}}

\def\bi{\begin{itemize}}
\def\ba{\begin{assumption}}
\def\ea{\end{assumption}}
\def\ee{\end{equation}}
\def\ei{\end{itemize}}

\def\xd{\dot{x}}
\def\Vd{\dot{V}}
\def\Vd{\dot{V}}
\def\ed{\dot{e}}

\def\cP{{\cal P}}
\def\cF{{\cal F}}

\def\xb{\bar{x}}

\def\cD{{\cal D}}

\def\cK{{\cal K}}

\def\cP{{\cal P}}

\def\cN{{\cal N}}

\def\cN{{\cal N}}
\def\bR{\mathbb{R}}

\def\cD{{\cal D}}
\def\cK{{\cal K}}

\def\bmat{\begin{matrix}}

\def\emat{\end{matrix}}

\ifCLASSINFOpdf
\else
\fi
\usepackage[tight,footnotesize]{subfigure}

\usepackage{url}

\begin{document}
\bstctlcite{IEEEexample:BSTcontrol}
	
%
\title{The distributed biased min-consensus protocol revisited: pre-specified finite time control strategies and small-gain based analysis}
%
%
%

\author{Yuanqiu Mo, \IEEEmembership{Member IEEE}, He Wang, \IEEEmembership{Member IEEE} 
\thanks{Mo and Wang are with the Department of System Science, School of Mathematics, Southeast University, Nanjing 211189, China (email: yuanqiumo@seu.edu.cn)}
}
\maketitle

\begin{abstract}
Unlike the classical distributed consensus protocols enabling the group of agents as a whole to reach an agreement regarding a certain quantity of interest in a distributed fashion, the distributed biased min-consensus protocol (DBMC) has been proven to generate advanced complexity pertaining to solving the shortest path problem. As such a protocol is commonly incorporated as the first step of a hierarchical architecture in real applications, e.g., robots path planning, management of dispersed computing services, an impedance limiting the application potential of DBMC lies in, the lack of results regarding to its convergence within a user-assigned time. In this paper, we first propose two control strategies ensuring the state error of DBMC decrease exactly to zero or a desired level manipulated by the user, respectively. To compensate the high feedback gains incurred by these two control strategies, this paper further investigates the nominal DBMC itself. By leveraging small gain based stability tools, this paper also proves the global exponential input-to-state stability of DBMC, outperforming its current stability results.  Simulations have been provided to validate the efficacy of our theoretical result.
\end{abstract}

\begin{IEEEkeywords}
Consensus, Biased Min-consensus, the Shortest Path Problem, Pre-specified Finite Time Control, Small Gain Theorems
\end{IEEEkeywords}

%
\IEEEpeerreviewmaketitle{\normalsize }

\section{Introduction}
The distributed consensus protocol, which aims to ensure that all agents reach an agreement in a distributed fashion, is one of the basic problems in the control field and has been extensively studied over the past decades \cite{yao2023event}. Depending on the agreement value, consensus can be roughly divided into 3 categories: 1) average-consensus; 2) min-consensus; and 3) max-consensus, their merits can be discovered in practical fields including formation control, clock synchronization, load balancing and so on \cite{yao2023event,min-max,cortesmaxmin}. Different from distributed consensus protocols which mainly serve as a means to mitigate state differences between agents/nodes, the distributed biased min-consensus protocol (DBMC) proposed in \cite{Zhang2017} has been proven to produce complex behaviors related to the shortest path problem \cite{bellman1958routing}, which is a complicated combinatorial optimization problems studied in computer science and artificial intelligence fields. In particular, DBMC is capable of enabling each node in a graph to find a shortest path (in the sense that the sum of weights of its constituent edges is minimized) from its nearest source. Moreover, compared with the classical centralized shortest path algorithms, e.g., Dijkstra's algorithm \cite{dijkstra2022note}, Bellman-ford algorithm \cite{bellman1958routing}, $\mathrm{A}^\ast$ algorithm \cite{russell2016artificial}, and etc, DBMC possesses better scalability and robustness in that 1) it is a distributed algorithm based on dynamic evolution; and 2) it achieves global stability as it has no requirement on the initial states while the aforementioned classical methods all require the initial states to be infinity.

Given the above advantages of DBMC, stability of DBMC (especially its discretized version) has been well studied from the control community. \cite{shi2018biased, zhang2017perturbing, yao2020hierarchical, shi2020distributed, TAC, mo2022stability} all have proven that the discrete-time DBMC and its variants will converge within a finite number of iterations while providing different upper bounds of such number. Our prior work also shows that the ultimate boundedness of the discrete-time DBMC under persistent perturbations on the edge weights, via a discrete-time Lyapunov-based approach. Owing to its finite time convergence, fruitful applications based on the discrete-time DBMC have exploited. To name a few, in \cite{shi2018biased, shi2020distributed} the discrete-time DBMC is adopted to search the fast charging station with the lowest overall objective, and to discover the best supply candidate with minimum power loss in power and traffic networks, respectively. \cite{yao2020hierarchical} utilizes the discrete-time DBMC to realize the optimal and safe path planning for unmanned surface vessels. In contrast to its discretized version, theoretical analyses of DBMC remains scarce. \cite{Zhang2017} proves the global asymptotic stability of DBMC via using LaSalle's invariance principle. The regional exponential stability (requiring the initial state errors to be non-negative) of DBMC is further demonstrated in \cite{mo2021lyapunov} by designing a non-smooth Lyapunov function. However, to the best of our knowledge, there is no existing work concerning convergence of DBMC within a prescribed time, which impedes the development of DBMC in real applications, as DBMC is commonly incorporated as a core and initial part of the hierarchical architecture in real applications, necessitating its convergence in prescribed time, e.g., obtaining the optimal path for the subsequent obstacles avoidance in robot path planning \cite{yao2020hierarchical}, or calculate the best route for the following management of dispersed computing services in content delivery networks \cite{paulos2019framework}. Even though the discrete-time DBMC can attain finite time convergence, the derived upper bounds on the finite time in the existing literature are all defined by structural parameters of the graph, such as the number of nodes or the diameter \cite{yao2020hierarchical, TAC, zhang2017perturbing}, resulting in its  inapplicability in situations where DBMC must complete within a given time, plus the incompatibility between discrete-time DBMC and dynamical systems, it is of great importance to develop DBMC such that convergence within a presetting time is guaranteed.  

Recently, the pre-specified finite time (PT) control for nonlinear systems has attracted increasing attention, the salient feature of such control strategy is that it ensures the system stabilizes within a finite presetting time  irrespective of initial conditions or any other design parameters \cite{song2023prescribed}, which outperforms the finite time or fix time control strategies, as their settling time either depends on the initial states or a conservative estimate \cite{ning2022fixed}. Moreover, the PT stabilizing effect can be achieved via placing a novel but simple time-varying scaling function in the feedback loop \cite{wang2018leader}, permitting it to be incorporated in DBMC to guarantee convergence in a prescribed time. However, when time tends to the presetting time infinite gains are inevitably incurred to achieve zero steady-state error via such control strategy \cite{krishnamurthy2020dynamic}, limiting the usage of this method beyond the presetting time interval while impeding its practical implementation simultaneously \cite{wang2018prescribed}. On this basis, the practical pre-specified finite time (PPT) control has been proposed \cite{cao2022practical}, which is a less ambitious version of PT control strategy in that, within the presetting time the state error will converge to the neighborhood of the origin, and the magnitude of the neighborhood can be adjusted by a user-defined parameter. Similar to the PT control strategy, practical finite time stabilizing effect can be attained by taking a continuous time-varying function, e.g., TBG, as the feedback gain \cite{ning2019practical}.

Inspired by the above PT and PPT control strategies and taking into account the urgent need in devising a DBMC equipped with user-defined convergence time, in this paper, we provide three design strategies for DBMC suitable for various application scenarios. As for the PPT stabilization, instead of using TBG directly, we repeatedly use the truncation of TBG such that the PPT stabilizing effect can be achieved orderly from near to far from the source till the entire network. In regard to the PT stabilization, technical issues arise as the feedback gain approaches to infinity at the presetting time, we obtain a PT stable solution beyond the presetting time interval, via the time-varying scaling function being concatenated with a constant function. To prove the existence of a continuous solution resulted from such a control strategy, a comparison function is subtly designed to prove the continuity of the dynamic of DBMC. Turning back to the nominal DBMC that adopts a constant feedback gain, upon which we have derived the range sufficient to guarantee the global exponential input-to-state stability (expISS) of DBMC, by leveraging small gain based approaches \cite{geiselhart2016relaxed,dashkovskiy2010small,ruffer2010monotone}. The contribution of the paper can be summarized as twofold:
\begin{itemize}
	\item To the best of our knowledge, there is no existing work addressing the pre-specified finite time convergence of the DBMC, which is an indispensable property for the implementation of DBMC in real applications, while this work fills the gap by proposing two control strategies such that PT and PPT stabilization on DBMC can be realized, respectively.
	\item As for DBMC itself, while existing works have demonstrated its global asymptotic stability or regional exponential stability, this work further proves its expISS, which not only implies the global exponential stability of DBMC without perturbations, but also reveals its state error eventually becomes small under small perturbations, regardless of the initial states.  
\end{itemize}
The rest of the paper is organized as follows. Section \ref{sec:pre} gives the necessary notations and definitions, as well as state the problem of interest. Section \ref{sec:ptandppt} introduces the PPT and PT control strategies on DBMC, respectively. Section \ref{sec:smallgain} introduces the small gain based analysis of the nominal DBMC. Section \ref{sec:simulations} provides the simulations, and Section \ref{sec:con} concludes.

\section{preliminary knowledge}\label{sec:pre}
As DBMC is a graph-based algorithm, in this section we first give the preliminary knowledge of graphs. The distributed biased min-consensus protocol, as well as some definitions facilitating our later stability analysis will also be presented.

We consider undirected graphs $G=(V,E)$ with $V = \{1,\cdots,n\}$ the set of nodes and $E$ the set of undirected edges. The set of neighbors of node $i \in V$ is denoted by $\cN(i)$, and no node is deemed to be its own neighbor, i.e., $i \notin \cN(i)$ for all $i \in V$. The weight of the edge between nodes $i$ and $j$ is denoted by $w_{ij}$ and is assumed to be positive,  $j \in \cN(i)$ implies $i \in \cN(j)$ and $w_{ij} = w_{ji}$ as the considered graph is undirected. Further, the set comprising the sources is denoted by $S \subsetneq V$.

In addition to the graphical framework introduced above, we also need the following notations, most of which will be used in Section \ref{sec:smallgain}. Define $\mathbb{R}, \mathbb{R}_+, \mathbb{Z}$ and $\mathbb{Z}_+$ as the set of real numbers, the set of nonnegative real numbers, the set of integers and the set of nonnegative integers, respectively. For any $x \in \mathbb{R}^n$, denote $|x|$, $|x|_{\infty}$ and $||x||$ as its $\ell_1$ norm, $\ell_\infty$ norm and any arbitrary fixed monotonic norm, respectively. For any function $\phi: \mathbb{R}_+ \rightarrow \bR^m$, its sup-norm is denoted by $||\phi||_\infty = \sup\{||\phi(t)||: t \in \bR_{+}\}$. The set of all functions $\mathbb{R}_+ \rightarrow \bR^m$ with finite sup-norm is denoted by $\ell^\infty(\bR^m)$. 

Comparison functions $\cK$ and $\cK_\infty$ are the same as defined in \cite{khalil}. 
We use the partial order induced by the component-wise ordering, i.e., for $x,y \in \bR_{+}^{n}$, we write $x < y~ (\mathrm{resp.}~x \leq y)$ to mean $x_i < y_i~(\mathrm{resp.}~x_i \leq y_i)$ for all $i \in \{1,\cdots,n\}$, and by $x \not\geq y$ we mean there exists at least one index $i$ such that $x_i < y_i$. A function $f: \mathbb{R} \rightarrow \mathbb{R}$ is of class $C^2$ if its second derivative is continuous. $co\{\cdot\}$ denotes the convex hull, i.e., for $n$ points $p_1, \cdots, p_n$, $co\{p_i~|~ i \in \{1,\cdots,n \}  \}= \{\sum_{j = 1}^{n}\lambda_jp_j: \lambda_j\geq 0 ~\mathrm{for ~all ~}j~ \mathrm{and} ~\sum_{j = 1}^{n}\lambda_j = 1 \}$.
\begin{figure}[h]
	\centering
	\includegraphics[width=0.35\columnwidth]{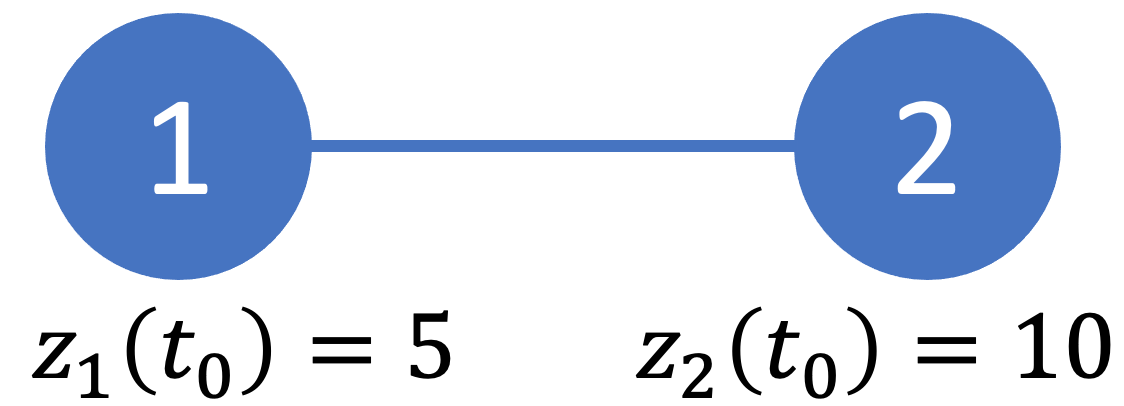}
	\caption{Illustration of the necessity of the non-increasing property of $z_i(\cdot)$. In this graph consisting of two nodes, if $z_i(\cdot)$ is not non-increasing, node 1 and node 2 will continuously exchange their initial value $z_1(t_0) = 5$ and $z_1(t_0) = 10$ and min-consensus will never be achieved.}
	\label{fig:example2}
\end{figure}

The definition of \emph{the shortest path problem}  considered in this paper is formally characterized as the follows. Note that in this paper multiple sources in the graph is allowed.
\begin{definition}
	Consider an undirected graph $G = (V,E)$, the shortest path problem aims to find a path from each node $i$ to its nearest source in $S$ such that the sum of the weights of its constituent edges is minimized.
\end{definition}
Let $x_i$ denote  the length/distance of the shortest path from node $i$ to the source set. According to the Bellman's principle of optimality \cite{bellman1958routing}, $x_i$ obeys the following recursion:
\begin{equation}\label{eq:bellman}
x_i = \begin{cases}
0, &i \in S \\
\min_{j \in \cN(i)} \{x_j +w_{ij} \},& i \notin S
\end{cases}.
\end{equation}

Before turning to the DBMC, we first revisit the rudiment of DBMC, the discrete-time min-consensus protocol, 
\begin{equation}\label{eq:min}
z_i(k+1) = \min_{j \in \cN(i)}\{ z_j(k) \}
\end{equation}
where $z_i(k)$ the state of node $i$ in the $k$-th time step of the discrete-time min-consensus protocol. Then it has been shown in \cite{min-max} that all nodes' states will converge to the minimum initial state within finite iterations if $z(\cdot)$ is a non-increasing function with respect to $k$. As shown in Figure \ref{fig:example2}, the non-increasing property is necessary as otherwise node $1$ and node $2$ will keep exchanging their initial states everlastingly and min-consensus will never be achieved. 

Assuming $t_0 = 0$ is the initial time, let $x_i(t)$ be the state (estimated distance of the shortest path from the source set at time $t$) of $i \in V$. Mimicking the style of (\ref{eq:min}) with using derivative instead of difference and adding suitable edge weights, DBMC proposed in \cite{Zhang2017, mo2021lyapunov} specifies the following recursion:
\begin{equation}\label{eq:practical-protocol}
\xd_i(t) = 
\begin{cases}
0, & i \in S\\
-\eta\left(x_i(t) - \min_{j \in \cN(i)}\{x_j(t)+w_{ij}\}\right), & i \notin S
\end{cases},
\end{equation}
where $x_i(0) = 0$ for all $i \in S$ and $\eta > 0$ is the feedback gain. The behavior of DBMC is also close to the recursion (\ref{eq:bellman}), in the sense that the state of the source node is anchored at 0 while that of a non-source node evolves using the minimum among the summations of its neighbors' states and the edge weights in between. 

It is worthwhile noting that the distributed biased min-consensus protocol is fully distributed as each node only receives 1) the current states from its neighbors and 2) the edge weight transmitted by its neighbors or measured by the node itself. It has been proven in \cite{Zhang2017} and \cite{mo2021lyapunov} respectively that DBMC can achieve global asymptotic stability and regional exponential stability (under a mild assumption that all initial states are overestimates, i.e., $x_i(0) \geq x_i$ for all $i \in V\setminus S$).

To facilitate the analysis in the next section, we introduce the following definitions based on (\ref{eq:bellman}).
\begin{definition}\label{def:true}
	A minimizing $j$ in the right hand side of the second bullet of (\ref{eq:bellman}) is called a true parent node of $i$. As $i$ may have multiple true parent nodes, we use $\cP(i)$ to denote the set of true parent nodes of node $i$. A source node does not have any true parent node.
	
	Further, for an undirected graph $G=(V,E)$, consider a sequence of nodes such that the predecessor of each node is one of its true parent nodes. We define $\cD(G)$, the effective diameter of $G$, as the longest length such a sequence can have in graph $G$.
\end{definition}
Another definition is needed to quantify the stability analysis.
\begin{definition}\label{def:longest}
	We call a path from a node $i$ to its nearest source $j \in S$ a shortest path, if it starts at $i$, ends with $j \in S$, and each node in the path is a true parent node of its predecessor. We call such a shortest path the longest shortest path if it has the most nodes among all shortest paths of $i$. The set $\cF_\ell$ is the set of nodes whose longest shortest paths to the source set have $\ell + 1$ nodes. 
\end{definition}
\begin{remark}
	Based on Definition \ref{def:longest}, there holds $\cF_0 = S, \cF_{\cD(G) + i} = \emptyset, \forall i \in \mathbb{Z}_+$ and $\cF_{i} \neq \emptyset, \forall i \in \{0,\cdots, \cD(G) - 1 \}$. Moreover, 
	\begin{equation}\label{eq:F}
		 \cup_{i \in \{0,\cdots, \cD(G) - 1  \}}\cF_i = V ~\mathrm{and}~ \cF_i \cap \cF_j = \emptyset, ~ \forall i \neq j,
	\end{equation}
(\ref{eq:F}) comes from the fact that each node $i \in V$ has a longest shortest path, and the number of nodes in such a path cannot exceed $\cD(G)$ defined in Definition \ref{def:true}. For any $i \in \cF_\ell$ with $\ell \in \{1,\cdots, \cD(G) - 1 \}$, it has a true parent node $j \in \cF_{\ell - 1}$.

 Besides, if $i$ has multiple paths comprising of different numbers of nodes, $i$ only belongs to  $\cF_\ell$ with $\ell+1$ the largest number of nodes  among those paths.
\end{remark}
\begin{figure}[h]
	\centering
	\includegraphics[width=1\columnwidth]{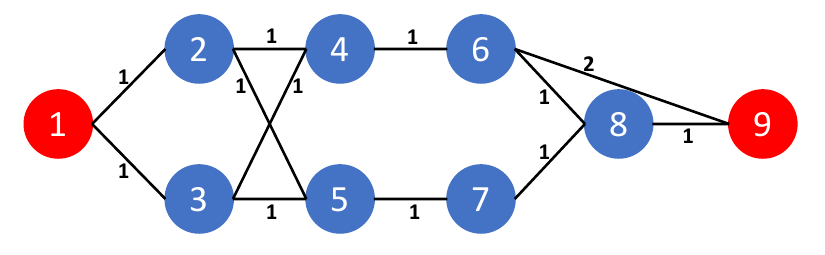}
	\caption{An undirected graph consisting of 7 non-source nodes and 2 source nodes. In this scenario, node 4 has 2 true parent nodes: node 2 and node 3. The effective diameter of the graph is 3. The longest shortest path of node 6 is $6 \rightarrow 8 \rightarrow 9$, and $6 \in \cF_2$ while $6 \notin \cF_1$.}
	\label{fig:example1}
\end{figure}
We use Figure \ref{fig:example1} to illustrate the above definition. In the above graph consisting of 9 nodes (2 source nodes, node 1 and node 9 are in red while 7 non-source nodes in blue), each edge weight is equal to 1 except that $w_{69} = 2$. The length of the shortest path from node 2 to its nearest source in the source set (node 1) is 2, resulted from the path $4 \rightarrow 2 \rightarrow 1$ or $4 \rightarrow 3 \rightarrow 1$, and thus the set of true parent nodes of node $4$ obeys $\cP(4) = \{2,3\}$. The effective diameter of the graph is 3, resulted from multiple sequences of nodes, e.g., $6 \rightarrow 8 \rightarrow 9$. 

Node 6 has two shortest paths towards its nearest source 9, $6 \rightarrow 8 \rightarrow 9$ and $6 \rightarrow 9$, while the path $6 \rightarrow 8 \rightarrow 9$ is the longest shortest path of node $6$. Further, $6 \in \cF_2$ while $6 \notin \cF_1$ by Definition \ref{def:longest}.

To avoid trivialities the main assumption throughout the paper is summarized as the follows.
\begin{assumption}\label{ass:main}
	The graph $G = (V, E)$ is connected and undirected, the source set $S \neq V$, each edge weight is positive, and $t_0 = 0$ is the initial time.
\end{assumption}

\section{pre-specified finite time control strategies}\label{sec:ptandppt}
In this section, we provide two control strategies on DBMC such that PPT or PT stabilization of (\ref{eq:protocol}) to the stationary value defined in (\ref{eq:bellman}) can be realized. 

The following assumption holds throughout this section.
\begin{assumption}\label{ass:overestimate}
	All non-source nodes have overestimated initial states, i.e., $x_i(0) \geq x_i$ for all $i \in V\setminus S$. 
\end{assumption}
\begin{remark}
	The above assumption is a mild one as one can set the initial states sufficiently large to ensure the assumption holds. Moreover, as will be shown later, the prescribed finite time for DBMC to converge is independent of the initial states under both the proposed two strategies.
\end{remark}
\subsection{The practical pre-specified finite time control strategy}\label{sec:ppt}
The PPT control strategy is realized via simply replacing the feedback gain in (\ref{eq:practical-protocol}) with the time base generator (TBG) gain \cite{ning2019practical, morasso1997computational}. Under TBG gain, (\ref{eq:protocol}) now is interpreted as 
\begin{equation}\label{eq:protocol}
\xd_i(t) = 
\begin{cases}
0, & i \in S\\
-\eta(t)\left(x_i(t) - \min_{j \in \cN(i)}\{x_j(t)+w_{ij}\}\right), & i \notin S
\end{cases}
\end{equation}
with the TBG gain $\eta(t)$ obeying 
\begin{equation}\label{eq:TBG}
\eta(t) = \frac{\dot{\varepsilon}(t - k*T_s)}{1 - \varepsilon(t - k*T_s) + \delta}, \forall t \in [k*T_s, (k+1)*T_s)
\end{equation}
where $k \in \mathbb{Z}_+$, $T_s$ is a pre-specified time constant,  $\varepsilon(t)$ is a TBG and $0 < \delta << 1$. In particular, $\varepsilon(t)$ has the following properties:
\begin{itemize}
	\item $\varepsilon(t)$ is at least $C^2$ on $(0, +\infty)$;
	\item $\varepsilon(t)$ is continuous and non-decreasing from an initial value $\varepsilon(0) = 0$ to a terminal value $\varepsilon(T_s) = 1$, where $T_s < +\infty$ is a prescribed time instant;
	\item $\dot{\varepsilon}(0) = \dot{{\varepsilon}}(T_s) = 0$;
	\item $\varepsilon(t) = 1$ when $t > T_s$.
\end{itemize}
A suitable example of TBG could be 
\begin{equation}\label{eq:tbgex}
\varepsilon(t) = \begin{cases}
\frac{10}{4^6}t^6 - \frac{24}{4^5} t^5+ \frac{15}{4^4}t^4, & 0\leq t \leq T_s\\
1, & t > T_s
\end{cases}.
\end{equation}
\begin{remark}\label{re:TBG}
Different from the TBG used in \cite{ning2019practical, morasso1997computational}, this paper indeed repeatedly uses a truncated version of TBG. As $\varepsilon(t)$ is of class $C^2$ on $(0, +\infty)$ and $\dot{\varepsilon}(T_s) = 0$, the TBG gain $\eta(t)$ used in (\ref{eq:protocol}) is continuous, bounded and non-negative on $[0, +\infty)$.
\end{remark}
As $\eta(t)$ in (\ref{eq:protocol}) is bounded and continuous on $[0,+\infty)$, (\ref{eq:protocol}) admits a unique solution per the following theorem.
\begin{theorem}\label{the:unique}
	Suppose Assumption \ref{ass:main} holds, (\ref{eq:protocol}) has a unique solution for $t \in [0, +\infty)$.
\end{theorem}

Before moving on, we introduce the following technical result regarding our TBG gain $\eta(t)$.
\begin{lemma}\label{le:practical}
Consider the differential equation
\begin{equation}\label{eq:TBGy}
	\dot{y}(t) = -\eta(t)y(t), y(0) = y_0, t \in [0, T_s]
\end{equation}
where $\eta(t)$ is defined in (\ref{eq:TBG}). Then $y(T_s) = \frac{\delta}{1 + \delta}y_0$ without dependency on the initial state.
\end{lemma}
Define $e_i(t)$ as the error between $x_i(t)$ the state of $i$ and its stationary value $x_i$ as 
\begin{equation}\label{eq:error}
e_i(t) = x_i(t) - x_i, \forall i \in V.
\end{equation}
With $x_i(t)$ defined in (\ref{eq:protocol}), magnitudes of the largest state error and the least state error are given by
\begin{equation}\label{eq:greatest}
V^+(t) :=  \max\bigg\{0,\underset{i \in V}{\max}\{e_i(t)\}\bigg\}
\end{equation}
and
\begin{equation}\label{eq:least}
V^-(t) :=  \max\bigg\{0,-\underset{i \in V}{\min}\{e_i(t)\}\bigg\},
\end{equation}
respectively. Apparently $V^+(t), V^-(t)\geq 0$ for all $t \geq 0$, and $V^+(t) = V^-(t) = 0$ implies $e_i(t) = 0$ for all $i \in V$, i.e., all states converge to their stationary values at time $t$. 

As $V^+(t)$ and $V^-(t)$ defined above are non-smooth functions, we need the following two sets to calculate their Clarke's generalized derivatives \cite{Zhang2017,clarke1990optimization}. Define $\cK(t)$ as the set comprising nodes of which state errors equal $V^+(t)$,
\begin{equation}\label{eq:K+}
\cK^+(t) = \{i\in V ~|~ e_i(t) = V^+(t) \}.
\end{equation}
Similarly, $\cK^-(t)$, the set comprising nodes whose state errors equal $V^+(t)$, is defined as
\begin{equation}\label{eq:K-}
\cK^-(t) = \{i\in V ~|~ e_i(t) = -V^-(t) \}.
\end{equation}
Both $V^+(t)$ and $V^-(t)$ are non-increasing, per the following two lemmas.
\begin{lemma}\label{le:greatest}
	Suppose Assumption \ref{ass:main} holds, with $\cK^+(t)$ defined in (\ref{eq:K+}), $V^+(t)$ defined in (\ref{eq:greatest}) obeys $\Vd^+(t) \leq 0$ for all $t \geq 0$.
\end{lemma}
\begin{proof}
According to the Clarke's generalized derivative \cite{Zhang2017,clarke1990optimization}	(detailed proof can be referred to sections 2.2 and 2.3 in \cite{clarke1990optimization}), $\Vd^+(t) = co\{\ed_i(t)~ | ~ i \in \cK^+(t)\}$ when $V^+(t) > 0$ and $\Vd^+(t) = co\{\ed_i(t), 0~ | ~ i \in  \cK^+(t) \}$ when $V^+(t) = 0$.

We need to consider three cases: 1) $\cK^+(t)\cap S = \emptyset$, i.e., $\exists i \in V\setminus S, e_i(t) > 0$; 2) $S = \cK^+(t)$, i.e., $\forall i \in V\setminus S$, $e_i(t) < 0$; 3) $S \subsetneq \cK^+(t)$, i.e., $\forall i \in V\setminus S$, $e_i(t) \leq 0$. 

In the first case, there holds $V^+(t) > 0$ as $\cK^+(t)\cap S = \emptyset$, it follows from (\ref{eq:protocol}) and (\ref{eq:error}) that
\begin{flalign}
&\Vd^+(t)  = \sum_{i \in \cK^+(t)}\lambda_i\eta(t)\big( -x_i(t) + \underset{k \in \cN(i)}{\min}~ \{x_k(t) + w_{ik} \}\big) \nonumber \\
& = \sum_{i \in \cK^+(t)}\lambda_i\eta(t)\big( -x_i - e_i(t) + \underset{k \in \cN(i)}{\min}~ \{x_k(t) + w_{ik} \}\big) \nonumber \\
& = \sum_{i \in \cK^+(t)}\lambda_i\eta(t)\big( -x_j - w_{ij} - e_i(t) + \underset{k \in \cN(i)}{\min}~ \{x_k(t) + w_{ik}\} \big) \label{eq:tr} \\
& \leq \sum_{i \in \cK^+(t)}\lambda_i\eta(t)\big( -x_j - w_{ij} - e_i(t) + x_j(t) + w_{ij}\big) \nonumber \\
& = \sum_{i \in \cK^+(t)}\lambda_i\eta(t)(e_j(t) - e_i(t)) \nonumber \\
& \leq 0 \label{eq:lar}
\end{flalign}
where in (\ref{eq:tr}) we assume $j$ is a true parent node of node $i$, and (\ref{eq:lar}) uses the fact that $\eta(t) > 0, \lambda_i \geq 0$ and $e_i(t) \geq e_j(t)$ for all $j \in V$ as $i \in \cK^+(t)$. 

In the second case, we have $V^+(t) = 0$ as $S = \cK^+(t)$, it follows from (\ref{eq:protocol}) and (\ref{eq:error}) that $\Vd^+(t) = \sum_{i \in S}\lambda_i\ed_i(t) = 0$ with $\lambda_i \geq 0$ and $\sum_{i \in S}\lambda_i = 1$ for $i \in S$.

For case 3), it follows from the above two cases that
\begin{flalign}
\Vd^+(t) = \sum_{i \in \cK^+(t)}\lambda_i\ed_i(t) &= \sum_{k \in S}\lambda_k\ed_k(t) + \sum_{j \in \cK^+(t)\setminus S}\lambda_j\ed_j(t) \nonumber \\
& \leq 0 
\end{flalign}
where $0 \leq \lambda_k, \lambda_j \leq 1$, and thus our proof is complete.
\end{proof}
To prove that $V^-(t)$ is also monotonically non-increasing, we first define $\cP_i(t)$, the set of current parent nodes of node $i$, as the follows:
\begin{equation}\label{eq:parent}
\cP_i(t) = \begin{cases}
\emptyset, & i \in S \\
\arg \underset{j \in \cN(i)}{\min} \{x_j(t) + w_{ij} \}, & i \notin S
\end{cases}.
\end{equation}
\begin{lemma}\label{le:least}
	Suppose Assumption \ref{ass:main}, with $\cK^-(t)$ defined in (\ref{eq:K-}), $V^-(t)$ defined in (\ref{eq:least}) obeys $\Vd^-(t) \leq 0$ for all $t \geq 0$.
\end{lemma}
\begin{proof}
	Again, from Clarke's generalized derivative, $\Vd^-(t) = co\{-\ed_i(t)~ | ~ i \in \cK^-(t)\}$ when $V^-(t) > 0$, and $\Vd^-(t) = co\{-\ed_i(t), 0~ | ~ i \in \cK^-(t)\}$ when $V^-(t) = 0$. 
	
	We consider the three cases: 1) $\cK^-(t)\cap S = \emptyset$, i.e., $V^-(t) > 0$; 2) $S = \cK^-(t)$; 3) $S \subsetneq \cK^-(t)$.  
	
	For the first case, it follows from (\ref{eq:protocol}) and (\ref{eq:error}) that
	\begin{flalign}
	\Vd^-(t) &= \sum_{i \in \cK^-(t)}\lambda_i\eta(t)\big( x_i(t) - \underset{k \in \cN(i)}{\min}~ \{x_k(t) + w_{ik} \}\big) \nonumber \\
	& = \sum_{i \in \cK^-(t)}\lambda_i\eta(t)\big( x_i + e_i(t) - x_k(t) - w_{ik}  \big) \label{eq:con}  \\
	& \leq \sum_{i \in \cK^-(t)}\lambda_i\eta(t)\big( x_k + w_{ik} + e_i(t) - x_k(t) - w_{ik}\big) \label{eq:us}  \\
	& = \sum_{i \in \cK^-(t)}\lambda_i\eta(t)(e_i(t) - e_k(t)) \nonumber \\
	& \leq 0 \label{eq:las}
	\end{flalign}
	where in (\ref{eq:con}) we assume $k \in \cP_i(t)$, (\ref{eq:us}) uses (\ref{eq:bellman}), and (\ref{eq:las}) uses the fact that $\eta(t) >0, \lambda_i \geq 0$ and $e_i(t) \leq e_j(t)$ for all $j \in V$ as $i \in \cK^-(t)$. 
	
	Then similar to the arguments in Lemma \ref{le:greatest}, there holds $\Vd^-(t) = 0$ in case 2) and $\Vd^-(t) \leq 0$ in case 3), completing our proof. 
\end{proof}
Lemma \ref{le:greatest} and Lemma \ref{le:least} indicate that both the largest overestimate and the least underestimate are non-increasing. Also,
note that both lemmas do not require Assumption \ref{ass:overestimate}. Under Assumption \ref{ass:overestimate}, we get the following appealing result.
\begin{lemma}\label{le:allover}
	Suppose Assumptions \ref{ass:main} and \ref{ass:overestimate} hold. With $x_i(t)$ and $x_i$ defined in (\ref{eq:protocol}) and \ref{eq:bellman}, respectively, there holds $x_i(t) \geq x_i$ for all $i \in V$ and $t \geq 0$. 
\end{lemma}
\begin{proof}
	In this case $V^-(0) = 0$. As $V^-(t) \geq 0$ by its definition and $\Vd^-(t) \leq 0$ by Lemma \ref{le:least}, we have $V^-(t) = 0$ for all $t \geq 0$, and thus our claim follows. 
\end{proof}
The following lemma characterizes the upper bound for the state error of (\ref{eq:protocol}). The basic pattern is that, once the state error of the true parent of some node is upper bounded at a given time, $T_s$ time later the node's own state error will drop below an upper bound. Both upper bounds are defined by the maximum initial state and the adjustable parameter $\delta$ in (\ref{eq:TBG}). To this end, we first define the maximum initial state error as
\begin{equation}\label{eq:maxerror}
e_{\max}(0) = \max_{i \in V} \{e_{i}(0)  \}.
\end{equation}
\begin{lemma}\label{le:bound}
	Suppose Assumptions \ref{ass:main} and \ref{ass:overestimate} hold, consider (\ref{eq:protocol}), with $e_i(t)$, $T_s$ and $\cF_\ell$ defined in (\ref{eq:error}), (\ref{eq:TBG}) and Definition \ref{def:longest}, respectively. Then for all $i \in \cF_\ell$ with $\ell \in \{1,\cdots, \cD(G)-1 \}$ there holds for all $t \geq \ell T_s$
	\begin{equation}\label{eq:instant}
		e_i(t) \leq \ell\frac{\delta}{1+\delta}e_{\max}(0),
	\end{equation}
where $0 < \delta << 1$ is introduced in (\ref{eq:TBG}).
\end{lemma}
As the lower bound of the state error has been determined in Lemma \ref{le:allover} under Assumption \ref{ass:overestimate}, the following theorem completes this subsection by demonstrating the PPT stabilization of (\ref{eq:protocol}).
\begin{theorem}\label{the:ppt}
	Suppose Assumptions \ref{ass:main} and \ref{ass:overestimate} hold, consider (\ref{eq:protocol}), with $e_i(t)$ and $T_s$ defined in (\ref{eq:error}) and Definition \ref{def:longest}, respectively. Then for all $i \in V$ there holds for all $t \geq (\cD(G) - 1)T_s$
	\begin{equation}\label{eq:practical}
	|e_i(t)| \leq \ell\frac{\delta}{1+\delta}e_{\max}(0),
	\end{equation}
	where $0 < \delta << 1$ is introduced in (\ref{eq:TBG}).
\end{theorem}
\begin{proof}
	This is a direct result of Lemma \ref{le:allover} and Lemma \ref{le:bound}.
\end{proof}
Under the PPT control strategy, each node's state will converge to the neighborhood of its stationary value within the settling time. Such settling time is prescribed without depending on the initial state. Though a state disagreement always exists after the settling time, one can adjust the parameter $\delta$ in the TBG to reduce the discrepancy to a desired level.

\subsection{The pre-specified finite time control strategy}\label{sec:pt}
A gap between the state and its stationary value always exists under the PPT control strategy  introduced in the previous subsection. To remedy such an imperfection, in this subsection the PT control strategy is presented such that the state error will diminish exactly to zero within a presribed time.

Similar to the PPT control strategy, the PT control strategy is achieved through replacing the feedback gain $\eta$ in (\ref{eq:practical-protocol}) with a time-varying scaling function $\bar{\eta}(t)$ such that (\ref{eq:practical-protocol}) becomes 
\begin{equation}\label{eq:prespecified}
	\xd_i(t) = 
	\begin{cases}
		0, & i \in S\\
		-\bar{\eta}(t)\left(x_i(t) - \min_{j \in \cN(i)}\{x_j(t)+w_{ij}\}\right), & i \notin S
	\end{cases}
\end{equation}
where 
\begin{equation}\label{eq:ebar}
\bar{\eta}(t) = 
\begin{cases}
\gamma + 2\frac{\dot{\rho}(t)}{\rho(t)}, & t \in [0, \bar{T}_s) \\
0 , & t \geq T_s
\end{cases}
\end{equation}
with $\gamma > 0$ and $\rho(t)$ obeying 
\begin{equation}\label{eq:rho}
	\rho(t) = 
		\frac{\bar{T}_s^{1 + h}}{(\bar{T}_s - t)^{1 + h}}, t \in [0, \bar{T}_s)
\end{equation}
where $\bar{T}_s$ is a prescribed time constant and $h > 1$ ($h \in \mathbb{Z}_+$) is user-chosen constant.
\begin{remark}\label{re:prespecified}
Unlike \cite{wang2018leader}, where the time-varying scaling function is constructed via 
reusing the first bullet of (\ref{eq:ebar}), 
the time-varying scaling function adopted in this paper reduces to zero constantly after the prescribed time. As will be shown later, this setup ensures the continuity of $\xd_i(t)$ and thus the existence of the solution to (\ref{eq:prespecified}).
\end{remark}
We first give a generalized result extended from Lemma 2 in \cite{wang2018leader}. 
\begin{lemma}\label{le:limi}
Consider the differential equation
\begin{equation}\label{eq:y}
\dot{y}(t) = -\left(\gamma + \alpha\frac{\dot{\rho}(t)}{\rho(t)} \right)y(t), y(0) = y_0, t \in [0, \bar{T}_s)
\end{equation}
where $\rho(t)$ is defined in (\ref{eq:rho}), $\gamma, \alpha > 0$. Then $y(t) = \rho^{-\alpha}(t)e^{-\gamma t}y_0$ for $t \in [0, \bar{T}_s)$.
\end{lemma}
By definition of $\rho(t)$, Lemma \ref{le:limi}
implies that
\begin{equation}
\lim_{t \rightarrow \bar{T}_s^-} y(t) = 0.
\end{equation}
Now let the state error $e_i(t)$, the greatest overestimate $V^+(t)$ and the least underestimate $V^-(t)$ be in the form of (\ref{eq:error}),  (\ref{eq:greatest}) and (\ref{eq:least}), respectively, with $x_i(t)$ there being replaced by the one defined in (\ref{eq:prespecified}). The following lemma summarizes their properties under the PT control strategy, over the interval $[0, \bar{T}_s)$.
\begin{lemma}\label{le:alloverpt}
	Suppose Assumption \ref{ass:main} holds, $V^+(t)$ and $V^-(t)$ defined above obey $\Vd^+(t) \leq 0$ and $\Vd^-(t) \leq 0$ for $t \in [0, \bar{T}_s)$. Moreover, suppose Assumption \ref{ass:overestimate} also holds, then $x_i(t) \geq x_i$ for all $i \in V$ and $t \in [0, \bar{T}_s)$.
\end{lemma}
\begin{proof}
	As $\bar{\eta}(t) \geq 1$ over $[0, \bar{T}_s)$, the proof follows directly from those of Lemma \ref{le:greatest} to Lemma \ref{le:allover}.
\end{proof}
With the lower bound of the state being determined in Lemma \ref{le:alloverpt}, the left-hand limit of each state at $\bar{T}_s$ is given in the following lemma, as the first step to calculate the left-hand limit of the derivative of each state.
\begin{lemma}\label{le:leftlimit}
	Suppose Assumptions \ref{ass:main} and \ref{ass:overestimate} hold, consider (\ref{eq:prespecified}), with $x_i$ defined in (\ref{eq:bellman}), for all $i \in V\setminus S$, there holds
	\begin{equation}\label{eq:leftlimit}
	\lim_{t \rightarrow \bar{T}_s^-} x_i(t) = x_i.
	\end{equation}
\end{lemma}
Though Lemma \ref{le:leftlimit} shows that each state converges to its stationary value when $t$ approaches the prescribed time $\bar{T}_s$, to demonstrate the finite time convergence of (\ref{eq:prespecified}) we still need to prove the continuity of each state on $[0, +\infty)$, a simple and feasible way is to show the continuity of its derivative, i.e., prove that the left-hand limit of $\xd_i(t)$ at $\bar{T}_s$ is 0. 

The following lemma first proves the continuity of derivatives of states for nodes in $\cF_1$, while furnishing the upper bound of state errors for the remaining nodes over $[0, \bar{T}_s)$.
\begin{lemma}\label{le:continuity}
	Suppose Assumptions \ref{ass:main} and \ref{ass:overestimate} hold, consider (\ref{eq:prespecified}), with $e_{\max}(0)$ the largest initial state error, $\rho(t)$ and $\cF_\ell$ defined in (\ref{eq:rho}) and Definition \ref{def:longest}, respectively, there holds
	\begin{equation}\label{eq:leftdot}
	\lim_{t \rightarrow \bar{T}_s^-} \xd_i(t) = \ed_i(t) = 0,~ \forall i \in \cF_1
	\end{equation}
and
\begin{equation}\label{eq:upper}
	e_i(t) \leq \frac{1}{\rho^{2(1 - \frac{1}{k}) }(t) }e^{-\gamma(1 - \frac{1}{k})t}ke_{\max}(0),~ \forall i \in \cF_\ell
\end{equation}
for $\ell \in \{2,\cdots, \cD(G) - 1 \}$ and some $k \geq 2$.
\end{lemma}
With the upper bound of state error provided in Lemma \ref{le:continuity}, we then use Squeeze theorem \cite{stewart2007multivariable} to prove the continuity of derivatives of states for the remaining nodes.
\begin{lemma}\label{le:squeeze}
	Suppose Assumptions \ref{ass:main} and \ref{ass:overestimate} hold, with $\cF_\ell$ defined in Definition \ref{def:longest}, there holds
	\begin{equation}\label{eq:remaining}
		\lim_{t \rightarrow \bar{T}_s^-} \xd_i(t) = \ed_i(t) = 0,~ \forall i \in \cF_\ell
	\end{equation}
for $\ell \in \{2,\cdots, \cD(G) - 1 \}$.
\end{lemma}
The pre-specified finite time convergence of (\ref{eq:prespecified}) comes as a byproduct of the above lemma.
\begin{theorem}\label{the:existence}
	Suppose Assumptions \ref{ass:main} and \ref{ass:overestimate} hold, consider (\ref{eq:prespecified}), with $x_i$ defined in (\ref{eq:bellman}), then for all $i \in V$, $x_i(t)$ is continuous on $[0, +\infty)$ and
	\begin{equation}
		x_i(t) = x_i, ~ \forall i \in V ~\mathrm{and}~ \forall t \geq \bar{T}_s.
	\end{equation}
\end{theorem}
\begin{proof}
The continuity of $x_i(t)$ on $[0, +\infty)$ is guaranteed by Lemma \ref{le:squeeze}, leading to $x_i(\bar{T}_s) = x_i$ for all $i \in V$ by (\ref{eq:leftlimit}) in Lemma \ref{le:leftlimit}. Further, as $\xd_i(t) = \ed_i(t) = 0$ for all $t \geq \bar{T}_s$, we have $x_i(t) = x_i$ for all $i \in V$ and $t \geq \bar{T}_s.$	
\end{proof}
The theorem below completes this section by showing that the solution to (\ref{eq:prespecified}) given in Theorem \ref{the:existence} is also unique, by exploiting the fact that each state lies in a compact set for $t \geq 0$.
\begin{theorem}
	Suppose Assumptions \ref{ass:main} and \ref{ass:overestimate} hold, (\ref{eq:prespecified}) admits a unique solution for $t \in [0, +\infty)$.
\end{theorem}
\begin{proof}
	From the proof of Theorem \ref{the:unique} and the fact that $\bar{\eta}(t)$ defined in (\ref{eq:ebar}) is continuous on $[0, \bar{T}_s)$, $f_i(t ,\tilde{x}) = -\bar{\eta}(t) (\tilde{x}_i - \min_{j \in \cN(i)} \{\tilde{x}_j + w_{ij} \})$ with $\tilde{x} = [\tilde{x}_1, \cdots, \tilde{x}_n]^\top \in \mathbb{R}^n$ and $i \in \{1,\cdots,n \}$, is Lipschitz continuous with respect to its second argument on $[0, t_1]$, with $t_1 < \bar{T}_s$. Further, as the largest state error is non-increasing by Lemma \ref{le:alloverpt} and $x_i(t)$ is continuous as well as $x_i(t) = x_i$ for $t \geq \bar{T}_s$ by Theorem \ref{the:existence}, define a compact set $P_i = [x_i, x_i + L]$ for all $i \in V$ with some $L > 0$, there holds $x_i(t)$ lies entirely in $P_i$ once $x_i(0) \in P_i$, since $f_i(t, \tilde{x})$ is continuous on $t$, it follows from Theorem 2.39 in \cite{haddad2008nonlinear} that the solution to (\ref{eq:prespecified}) is unique.
\end{proof}

\section{Small gain based analysis}\label{sec:smallgain}
In this section, we return to the nominal model of the distributed biased min-consensus protocol (\ref{eq:practical-protocol}). We will show that (\ref{eq:practical-protocol}) is globally exponentially input-to-state stable (per Definition X) while its edge weight is under persistent perturbation, i.e., the edge weight $w_{ij}$ in (\ref{eq:practical-protocol}) is perturbed from its nominal value and becomes time-varying. Under such a perturbation (\ref{eq:practical-protocol}) now turns to  
\begin{equation}\label{eq:per-protocol}
	\xd_i(t) = 
	\begin{cases}
		0, & i \in S\\
		-\eta\left(x_i(t) - \min_{j \in \cN(i)}\{x_j(t)+w_{ij}(t)\}\right), & i \notin S
	\end{cases}.
\end{equation}
Unlike \cite{mo2021lyapunov}, where the perturbation on the edge weight is assumed to be additive, i.e., $w_{ij}(t) = w_{ij} + \epsilon_{ij}(t)$ with $\epsilon_{ij}(t) > 0$, here the only assumption  is that, under such perturbation the edge weight is bounded above 0, i.e.,
\begin{equation}\label{eq:boundnoise}
	0 < w_{\min} \leq w_{ij}(t) \leq w_{\max}, \forall t \geq 0, i \in V, j \in\cN(i).
\end{equation}
 Note that (\ref{eq:boundnoise}) allows asymmetric perturbation, i.e., $w_{ij}(t) \neq w_{ji}(t)$ is allowed. 
 
Before the stability analysis, we first translate (\ref{eq:per-protocol}) into a nonlinear dynamical system. We continue to adopt the notions of state error $e_i(t)$, and node $i$'s set of parent nodes $\cP(i)$ as defined in Section \ref{sec:pt}, with $x_i(t)$ now being defined as in (\ref{eq:per-protocol}).  We further define 
\begin{equation}\label{eq:input}
	u_{ij}(t) = w_{ij}(t) - w_{ij}
\end{equation}
as the deviation of the edge weight from is nominal value. Then (\ref{eq:per-protocol}) can be rewritten as 
\begin{flalign}
&\ed_i(t) = 0, ~\mathrm{if}~ i \in S \nonumber \\
&	\ed_i(t) = -\eta\big(x_i + e_i(t) - \nonumber \\
& ~~~~~~~~\min_{j \in \cN(i)}\{e_j(t)+ x_{j} + u_{ij}(t) + w_{ij}\}\big), ~\mathrm{if}~ i \notin S \label{eq:tran}
\end{flalign}
with $x_i, x_j$ and $w_{ij}$ the structural aspects the graph $G$. By taking $e_i(t)$ and $u_{ij}(t)$ as the state and the input, respectively, (\ref{eq:tran}) can be further described by a nonlinear map $f_i: \mathbb{R}^{|\cN(i)|+1} \times \mathbb{R}^{|\cN(i)|} \rightarrow \mathbb{R}$, i.e.,
\begin{equation}\label{eq:subsystem}
\ed_i(t) = f_i(e_i(t),e_{i_1}(t),\cdots,e_{i_{|\cN(i)|}}(t), u_i(t) ) 
\end{equation}
where $i_\ell\in \cN(i)$ for $\ell \in \{1,\cdots, |\cN(i)|  \}$ and $u_i(t) = [u_{i,i_1}(t), \cdots, u_{i,i_{|\cN(i)|}}(t)]^\top \in \mathbb{R}^{|\cN(i)|}$. The overall system can be defined by the following composite nonlinear map $F: \mathbb{R}^{n}\times \mathbb{R}^{2|E|} \rightarrow \mathbb{R}^{n}$
\begin{equation}\label{eq:overall}
\ed(t) = F(e(t), u(t))
\end{equation}
where $e(t) = [e_1(t), \cdots, e_n(t)] \in \mathbb{R}^{n}$ and $u(t) = \left( u_{ij}(t) \right)_{i\in V, j\in\cN(i)} \in \mathbb{R}^{2|E|}$. Also, $F(0,0) = 0$ as in this case $\ed_i(t) = e_i(t) = 0$ for all $i \in V$ by (\ref{eq:bellman}).

It can be seen from (\ref{eq:tran}), (\ref{eq:subsystem}) and (\ref{eq:overall}) that $F$ is continuous on $\mathbb{R}^{n}\times \mathbb{R}^{2|E|}$. Further, as $\eta > 0$ is a constant feedback gain, it follows from the proof of Theorem \ref{the:unique} that $F$ is globally Lipschitz continuous with respect to its first argument. Based on those facts, Assumption 2.1 in \cite{Mironchenko2023} holds, and thus global exponential input-to-state stability can be defined for (\ref{eq:overall}), with details given below. 
\begin{definition}\label{def:expISS}
\cite{Mironchenko2023} (\ref{eq:overall}) is said to be globally exponentially input-to-state stable (expISS) if 
\begin{itemize}
	\item (\ref{eq:overall}) is forward complete, i.e., for all the initial state $e(0) \in \mathbb{R}^n$ and all the input $u \in \ell^\infty(\mathbb{R}^{2|E|})$, the corresponding solution to (\ref{eq:overall}) exists and is unique on $[0,+\infty)$;
	\item  there exist $c, p > 0$ and $\lambda_u \in \cK$ such that, for all the initial state $e(0) \in \mathbb{R}^n$, $u \in \ell^\infty(\mathbb{R}^{2|E|})$ the following holds
	\begin{equation}\label{eq:expISS}
		|| e(t) || \leq ce^{-p t}|| e(0) || + \lambda_u(|| u ||_\infty),~\forall t \geq 0.
	\end{equation}
\end{itemize}
\end{definition}
Obviously (\ref{eq:overall}) is globally exponentially stable (per Definition 4.5 in \cite{khalil}) with 0-input once it is expISS. A direct method to demonstrate the expISS of (\ref{eq:overall}) is to show the existence of an exponential ISS Lyapunov function (refer to Definition 6.28 and Proposition 6.29 in \cite{Mironchenko2023} for more details), which is central in this section.
\begin{definition}\label{def:eISS}
	A continuous function $V: \mathbb{R}^n \rightarrow \mathbb{R}_+$ is called an exponential ISS Lyapunov function (eISS Lyapunov function) for (\ref{eq:overall}) if there exist constants $\underline{\omega}, \bar{\omega}, b, \kappa > 0$ and $\chi \in \cK_\infty$ such that for all $x \in \mathbb{R}^n$ and all $u \in \ell_\infty(\mathbb{R}^{2|E|})$
	\begin{equation}\label{eq:radial}
	\underline{\omega}|| x ||^b \leq	V(x) \leq \bar{\omega}|| x ||^b,
	\end{equation}
	\begin{flalign}
	& 	V(x) \geq \chi(|| u ||_\infty) \implies \forall \xi \in co\{\bar{\xi}: \exists x_k \rightarrow x: \frac{dV}{dx}\rightarrow \bar{\xi}   \}  \nonumber \\
	&\Vd(x) := <\xi, \xd(t)> \leq -\kappa V(x).\label{eq:imp}
	\end{flalign}
	where $<\cdot, \cdot>$ denotes the standard scalar product.
\end{definition} 
\begin{remark}\label{re:impdissi}
It should be noted that the upper right-hand Dini derivative is used in \cite{Mironchenko2023}, while here we use the Clarke's generalized derivative for (\ref{eq:imp}) instead, in this paper these two derivatives remain to be same as we will adopt $\ell_1$ norm $|\cdot|$  as our candidate eISS Lyapunov function.
	
	The above eISS Lyapunov function is called an implication-form eISS Lyapunov function. Another candidate to guarantee expISS is the dissipative-form eISS Lyapunov function, which satisfies (\ref{eq:radial}) with (\ref{eq:imp}) being replaced by $\Vd(x) \leq -\epsilon V(x) + \gamma(|| u ||_\infty)$ for some $\epsilon > 0$ and $\gamma \in \cK_\infty$. It has been shown in Proposition 2.11 of \cite{Mironchenko2023} that the dissipative-form eISS Lyapunov function implies the implication-form one, and both of them imply the expISS of (\ref{eq:per-protocol}).
\end{remark}
We summarize the assumption throughout this section.
\begin{assumption}\label{ass:ltz}
	Each edge weight is perturbed from its nominal value and obeys (\ref{eq:boundnoise}). All initial states are non-negative.
\end{assumption}
\begin{remark}\label{re:ltz}
	Unlike the requirements that all states need to be overestimates as assumed in the previous section or in \cite{mo2021lyapunov}, the non-negative restriction on the initial states given above is only used for facilitating the stability analysis. Indeed, by using absolute value of received states from neighbors, i.e., (\ref{eq:per-protocol}) becomes 
	\begin{equation}\label{eq:abs-protocol}
		\xd_i(t) = 
		\begin{cases}
			0, & i \in S\\
			-\eta\left(x_i(t) - \min_{j \in \cN(i)}\{|x_j(t)|+w_{ij}(t)\}\right), & i \notin S
		\end{cases},
	\end{equation}
all states will become non-negative after a finite time, and (\ref{eq:abs-protocol}) reduces to (\ref{eq:per-protocol}) thereafter. Proof of this claim is provided in the Appendix.
\end{remark}

The key steps of our proof consists of two parts, we will first show that the state error of each node behaves close to the implication-form eISS Lyapunov function given in Definition \ref{def:eISS}. Then we prove the expISS of (\ref{eq:per-protocol}), via showing the existence of an implication-form eISS Lyapunov function, which is constructed by imposing a cyclic small gain condition \cite{dashkovskiy2010small} on the state error trajectory of each node.

Recall the stationary value defined in (\ref{eq:bellman}), based on the fact that numbers of nodes and edges are both finite in graph $G$, there must exist a $\zeta \in (0,1)$ such that for all $i \in V$ with $\mathcal{N}(i)\setminus\mathcal{P}(i) \neq \emptyset$
\begin{equation}\label{eq:zeta1}
	\frac{ x_i }{ x_l + w_{il} } \leq \zeta, ~ \forall l \in \mathcal{N}(i)\setminus\mathcal{P}(i).
\end{equation}
We exemplify (\ref{eq:zeta1}) using Figure \ref{fig:example1}, for node 8, we have $\cN(8) = \{6,7,9\}$ while $\cP(8) = \{9\}$, and there holds $\frac{x_8}{x_6 + w_{68}} = \frac{x_8}{x_7 + w_{78}} = \frac{1}{3}$, and it can be verified that (\ref{eq:zeta1}) holds for any node $i$ satisfying $\mathcal{N}(i)\setminus\mathcal{P}(i) \neq \emptyset$ in Figure \ref{fig:example1} with $\zeta = \frac{2}{3}$. 

We first show that the state error of each node admits an eISS Lyapunov-like trajectory. 
\begin{lemma}\label{le:trajectory}
	Suppose Assumptions \ref{ass:main} and \ref{ass:ltz} hold, consider (\ref{eq:per-protocol}), with $\zeta$ defined in (\ref{eq:zeta1}), $e_i(t)$ and $\cP_i(t)$ the state error and the set of current parent nodes of node $i$ under (\ref{eq:per-protocol}), respectively, let $V_i(\cdot) = |\cdot|$ for $i \in V$, then for $t \geq 0$
	\begin{equation}\label{eq:fors}
		V_i(e_i(t)) = 0, ~ \forall i \in S,
	\end{equation}
	\begin{flalign}\label{eq:forns}
&	V_i(e_i(t)) \geq  \lambda_{ij} V_j(e_j(t)) + \lambda_{iu} || u ||_\infty \implies  \nonumber \\
& \Vd_i(e_i(t)) \leq -(\eta - 1)V_i(e_i(t)),~ \forall i \in V\setminus S
	\end{flalign}
	where in (\ref{eq:forns}) either $\lambda_{ij} = \lambda_{iu} = \eta$ with $j$ the true parent of $i$, or $\lambda_{ij} = \zeta\eta, \lambda_{iu} = \eta$ with $j \in \cP_i(t)$.  
\end{lemma}
With state error trajectories characterized in (\ref{eq:fors}) and (\ref{eq:forns}), we define the $n \times n$ gain matrix $\Gamma: (\bar{\lambda}_{ij})_{i,j \in \{1,\cdots, n \}}$ with the element $\bar{\lambda}_{ij}$ in the $i$-th row and $j$-th column obeying 
\begin{equation}\label{eq:gain}
	\bar{\lambda}_{ij} = \begin{cases}
		\zeta\eta,& i\notin S, j \in \cP(i)\\
		\eta, &  i\notin S, j \notin \cP(i)\\
		0.5/\eta^{\cD(G) - 1}, & i \in S
	\end{cases}
\end{equation} 
where $\cP(i)$ and $\cD(G)$ in (\ref{eq:gain}) are the set of true parent nodes of node $i$ and the effective diameter of graph $G$, respectively. The corresponding $n \times n$ adjacency matrix $A_\Gamma: (a_{ij})_{i,j \in \{1,\cdots, n\}}$ is defined by 
\begin{equation}\label{eq:A}
a_{ij} = 0~ \mathrm{if} ~ \bar{\lambda}_{ij} = 0, ~\mathrm{and}~  a_{ij} = 0~ \mathrm{otherwise}
\end{equation}
Obviously it follows from (\ref{eq:fors}) and (\ref{eq:forns}) that for all $i \in V$
\begin{flalign}\label{eq:tra}
	&	V_i(e_i(t)) \geq \max\{  \bar{\lambda}_{i1} V_1(e_1(t)), \cdots,  \bar{\lambda}_{in} V_1(e_1(t)) \}+\eta || u ||_\infty   \nonumber \\
	&\implies \Vd_i(e_i(t)) \leq -(\eta - 1)V_i(e_i(t))
\end{flalign}

A gain matrix $\Gamma$ is called irreducible if and only if the graph resulted from its  corresponding ajacency matrix is strongly connected \cite{dashkovskiy2010small,berman1994nonnegative}. Then as $A_\Gamma$ has no zero elements by (\ref{eq:gain}), the gain matrix defined by (\ref{eq:gain}) is irreducible. 

Define the following map motivated by the gain matrix $\Gamma$
\begin{gather}\label{eq:gmd}
	\Gamma_{\bigoplus}: \bR_{+}^{n} \rightarrow \bR_{+}^{n}, \begin{bmatrix}
		s_1  \\
		\vdots  \\
		s_{n}
	\end{bmatrix} \mapsto \begin{bmatrix}
		\underset{k \in \{1,\cdots,n\}}{\max}\{\bar{\lambda}_{1k}s_k\} \\
		\vdots \\
		\underset{k \in \{1,\cdots,n\}}{\max}\{\bar{\lambda}_{n k}s_k\}
	\end{bmatrix},
\end{gather}
as well as the diagonal operator $D: \bR_+^n \rightarrow \bR_+^n$ by 
\begin{equation}\label{eq:D}
D: \bR_{+}^{n} \rightarrow \bR_{+}^{n}, [s_1, \cdots, s_n]^\top \mapsto [ds_1, \cdots, ds_n]^\top
\end{equation}
with the diagonal operator factor $d > 1$.

The following lemma gives a sufficient condition for the existence of a diagonal operator $D$ such that $D\left(\Gamma_{\bigoplus}(x)\right) \not\geq x$ for all $x \in \bR_{+}^{n}$, which is kernal of later small gain condition.
\begin{lemma}\label{le:D}
	Suppose Assumptions \ref{ass:main} and \ref{ass:ltz} hold, consider the gain matrix, the gain matrix induced map and the diagonal operator $D$ defined in (\ref{eq:gain}), (\ref{eq:gmd}) and (\ref{eq:D}), respectively, with $\eta$ and $\zeta$ defined in (\ref{eq:per-protocol}) and (\ref{eq:zeta1}), respectively, if $1 < \eta < 1/\zeta$, then there exists a diagonal operator $D$ such that $D\left(\Gamma_{\bigoplus}(x)\right) \not\geq x$ for all $x \in \bR_{+}^{n}$.
\end{lemma}
Now we are ready the provide the small gain based theorem that manifests the expISS of (\ref{eq:per-protocol}). 
\begin{theorem}\label{th:small}
	Suppose Assumptions \ref{ass:main} and \ref{ass:ltz} hold, consider (\ref{eq:per-protocol}), with $e_i(t)$ the state error of node $i$, $u_{ij}(t)$, $\eta$ and $\zeta$ defined in (\ref{eq:input}), (\ref{eq:per-protocol}) and (\ref{eq:zeta1}), respectively. If $1 < \eta < 1/\zeta$, then (\ref{eq:per-protocol}) admits  an eISS Lyapunov function defined in Definition \ref{def:eISS}, using $e(t)$ the state error vector as the argument and $u(t) = (u_{ij})_{i\in V, j \in \cN(i)}$ as the input. 
\end{theorem}
\begin{proof}
The proof mainly utilizes Theorem 5.3 and Corollary 5.6 in \cite{dashkovskiy2010small}. From Lemma \ref{le:trajectory}, $V_i = |\cdot|$ for node $i$ is indeed an ISS Lyapunov function as defined in Definition 2.5 in \cite{dashkovskiy2010small}. Further, as $1 < \eta < 1/\zeta$, Lemma \ref{le:D} holds, combining the fact that $\bar{\lambda}_{ij}$ in $\Gamma$ defined in (\ref{eq:gain}) as well as  the diagonal operator factor $d$ defined in (\ref{eq:D}) are both constants, it follows from Corollary 5.7 in \cite{geiselhart2016relaxed} and \cite{geiselhart2012numerical} that there exists a vector of linear functions $\sigma = [\sigma_1, \cdots, \sigma_n]^\top \in \cK_\infty^n$ such that $D\left(\Gamma_{\bigoplus}(\sigma(s)) \right) < \sigma(s)$ for all  $s > 0$.  

Utilizing the above $\sigma$, by Theorem 5.3 in \cite{dashkovskiy2010small}, the eISS Lyapunov function of (\ref{eq:per-protocol}) can be constructed as 
\begin{equation}\label{eq:V}
	V(e(t)) = \max_{i \in \{1, \cdots, n \}  }\sigma_i^{-1}V_i(e_i(t))
\end{equation}
where $\sigma_i$ obeys $\sigma_i^{-1}(\sigma_i(s)) = s$ for all $s \geq 0, i \in \{1\cdots,  n\}$ and $\sigma^{-1}$ is also linear. As $V_i(\cdot) = |\cdot|$ by Lemma \ref{le:trajectory}, (\ref{eq:radial}) holds for $V$ defined in (\ref{eq:V}). Further, as $\eta - 1 > 0$ is a constant, $\alpha$ in equation (5.9) of Theorem 5.3 in \cite{dashkovskiy2010small}  is also linear, indicating that (\ref{eq:imp}) also holds for $V$ defined in (\ref{eq:V}). Therefore, $V$ is an implication-form eISS Lyapunov function for (\ref{eq:per-protocol}) and thus (\ref{eq:per-protocol}) is expISS.
\end{proof}

\section{Simulations}\label{sec:simulations}
In this section, we empirically verify the theoretical results presented in the previous sections. Unless otherwise mentioned, all simulations are implemented on a $4 \times 1$ km field consisting of 500 nodes, including one source. All the nodes are randomly distributed and communicating within a 0.25km.
\subsection{The practical pre-specified finite time control strategy}

\begin{figure}
	\centering
	\subfigure[Randomized graph with $\cD(G) = 20$]{
		\includegraphics[width = 0.9\columnwidth]{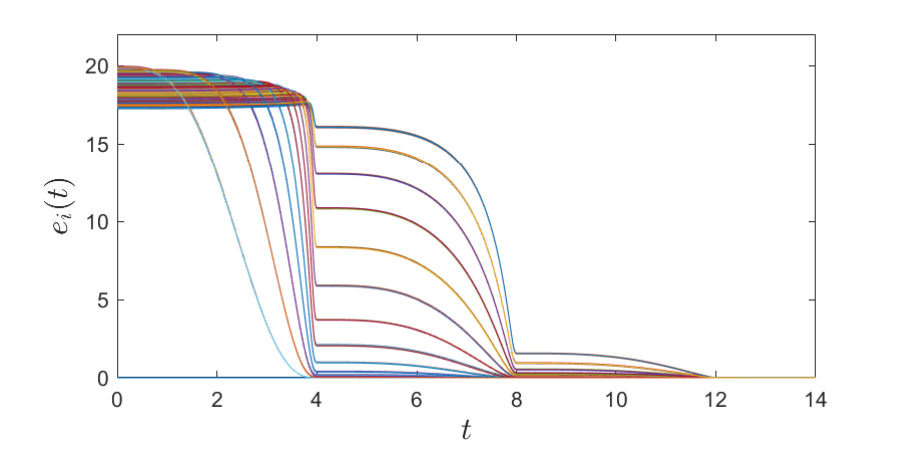}}
	\subfigure[Randomized graph with $\cD(G) = 22$]{  
		\includegraphics[width = 0.9\columnwidth]{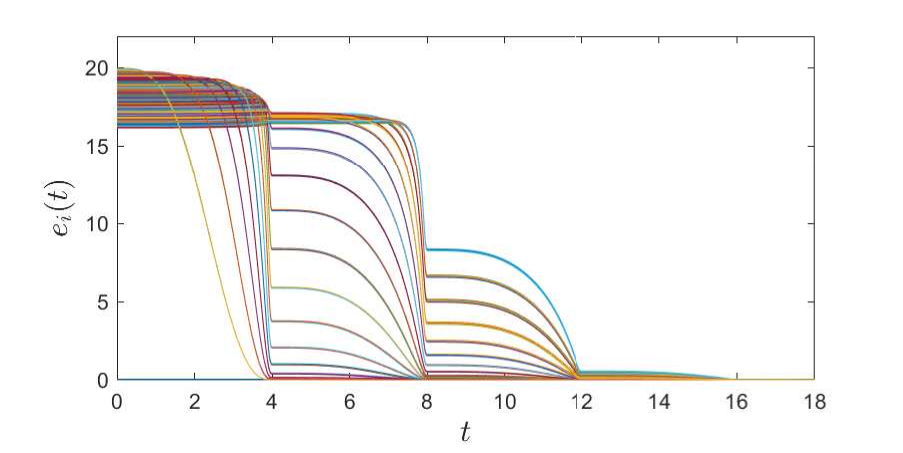}}
	\subfigure[Line graph with $n = \cD(G) = 50$]{  
		\includegraphics[width = 0.9\columnwidth]{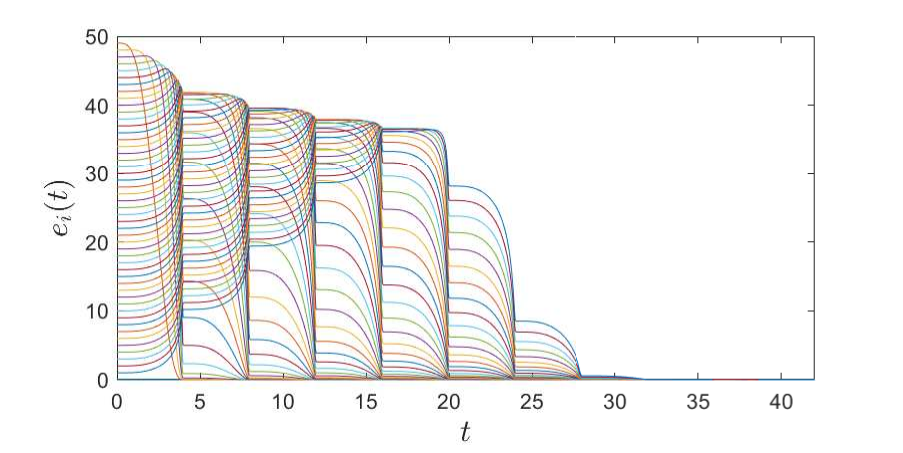}}
	\subfigure[Line graph with $n = \cD(G) = 100$]{  
		\includegraphics[width = 0.9\columnwidth]{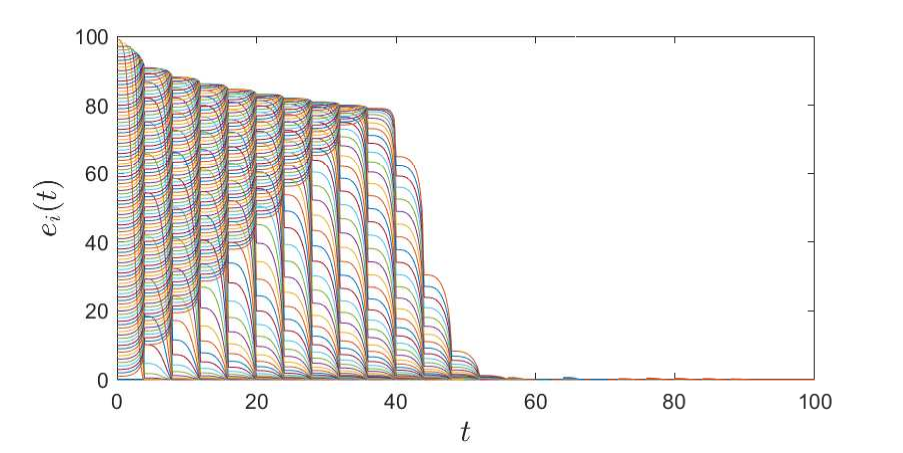}}
	\caption{Results of applying DBMC using the practical pre-specified finite time control strategy to randomized graphs and line graphs. In all scenarios, state errors drop below the bound given in Theorem \ref{the:ppt}. It takes much longer time for line graphs than randomized graphs for the state errors to be bounded, while the time needed for all the cases are less than the theoretical bound provided in Theorem \ref{the:ppt}.}
	\label{fig:ppt}
\end{figure}
We first adopt TBG $\varepsilon(t)$ defined in (\ref{eq:tbgex}) to build our TBG gain $\eta(t)$ as in (\ref{eq:TBG}). In this case $T_s = 4$ as $\varepsilon(4) = 1$. By setting $\delta$ in (\ref{eq:TBG}) as $\delta = 10^{-4}$, the simulation results are shown in Figure \ref{fig:ppt} (a) and (b). In this scenario, all initial states are set as overestimates such that Assumption \ref{ass:overestimate} holds. In Figure \ref{fig:ppt} (a), the effective diameter $\cD(G) = 20$, the state error $e_i(t)$ of each node $i \in V$ drops rapidly below the bound furnished in Theorem \ref{the:ppt} within $3 * T_s = 12$ units of time, the time needed for the state errors to be bounded is much less than the theoretical bound $(\cD(G) - 1)T_s$ units of time. Similarly, in Figure \ref{fig:ppt}, such time is $3 * T_s = 12$ units of time while $\cD(G)$ in this scenario is 22. A reasonable explanation could be the greatest overestimate $V^+(t)$, i.e., the largest state error, is non-increasing itself by Lemma \ref{le:greatest}. Also, the graph structural aspects will also play a role in the time guaranteeing all state errors to be bounded. As shown in Figure \ref{fig:ppt} (c) and (d), while implementing the protocol on line graphs with 50 and 100 nodes (in both line graphs the source node is located at the rightmost side such that the number of nodes is equal to the effective diameter), respectively, it takes much longer than randomized graphs for the state errors to be bounded, even though line graphs have fewer nodes. In the line graph with 50 nodes, $10 * T_s = 40$ units of time is needed for the state errors dropping below the bound, while $22 * T_s = 88$ units of time is needed for the other line graph.

\subsection{The pre-specified finite time control strategy}
\begin{figure}
	\centering
	\subfigure[Randomized graph with $\bar{T}_s = 2, h = 2, \gamma = 1$]{
		\includegraphics[width = 0.9\columnwidth]{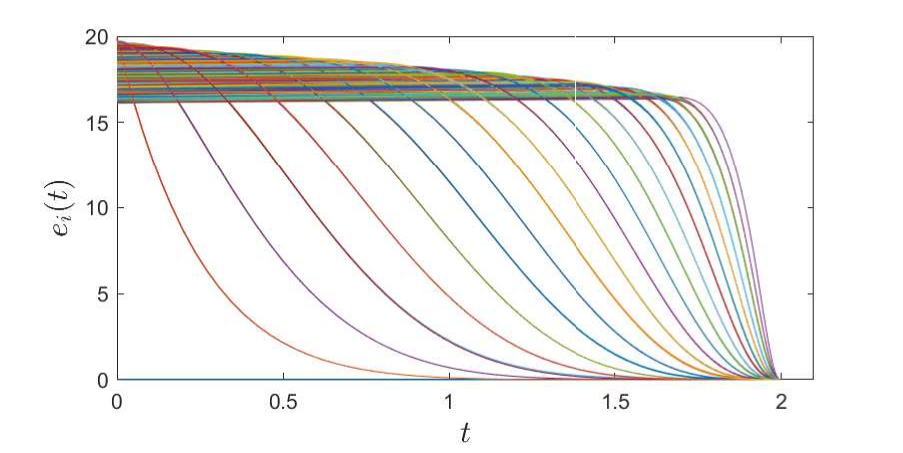}}
	\subfigure[Randomized graph with $\bar{T}_s = 4, h = 2, \gamma = 1$]{  
		\includegraphics[width = 0.9\columnwidth]{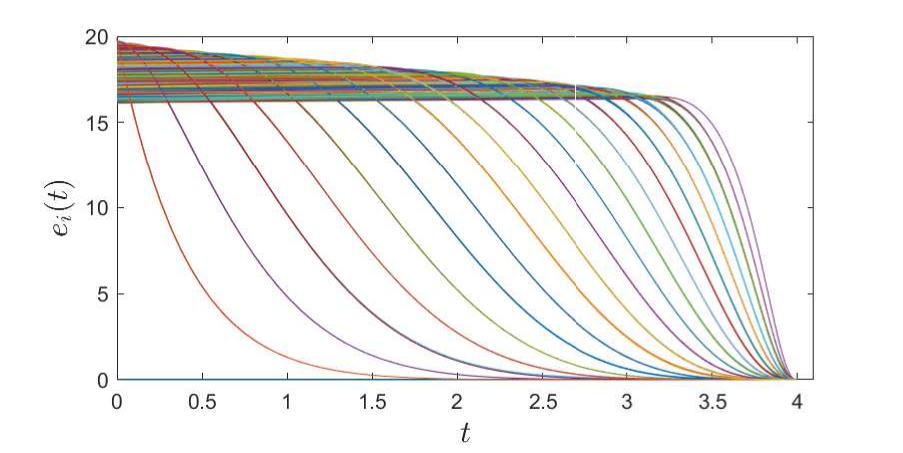}}
	\subfigure[Randomized graph with $\bar{T}_s = 4, h = 2, \gamma = 10$]{  
		\includegraphics[width = 0.9\columnwidth]{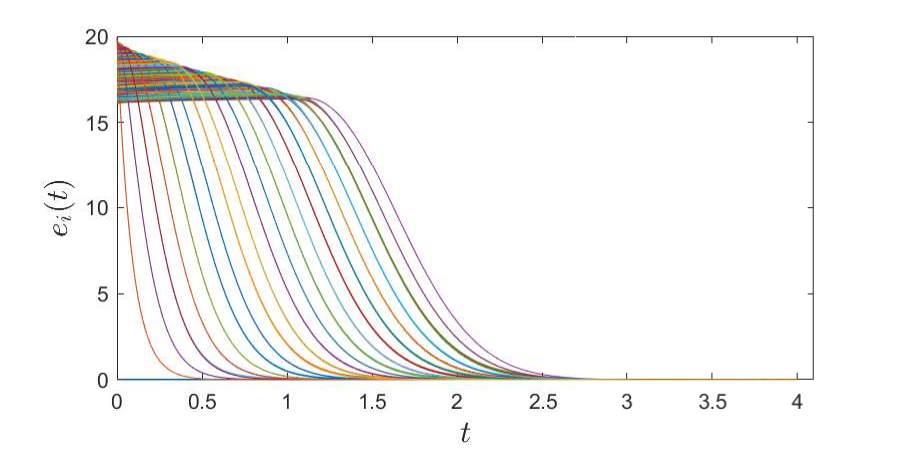}}
	\subfigure[Randomized graph with $\bar{T}_s = 4, h = 20, \gamma = 10$]{  
		\includegraphics[width = 0.9\columnwidth]{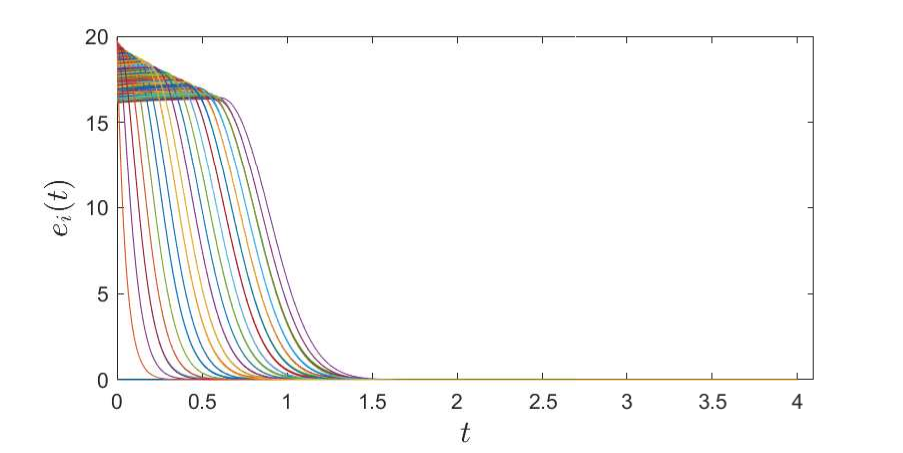}}
	\caption{Results of applying DBMC using the pre-specified finite time control strategy to randomized graphs. In all scenarios the state errors converge to 0 within the prescribed time. The protocol achieves a better convergence effect with a larger $h$ or $\gamma$.}
	\label{fig:pt}
\end{figure}
We next investigate the performance of the pre-specified finite time control strategy. For $\bar{\eta}(t)$ defined in (\ref{eq:ebar}), we first fix $h = 2$ and $\gamma = 1$ while varying $\bar{T}_s$, which is the prescribed time for the convergence of DBMC. It has been shown in Figure \ref{fig:pt} (a) and (b) that for a given randomized graph with $\cD(G) = 23$, all state errors converge within the prescribed time $\bar{T}_s = 2$ and $\bar{T}_s = 4$ units of time, respectively. Further, we fix $h = 2$ and $\bar{T}_s = 4$ while increasing $\gamma$ from 1 to 10, i.e., increasing the feedback gain, the result is shown in Figure \ref{fig:pt} (c), in this case all state errors still converge within $\bar{T}_s = 4$ units of time, but comparing with Figure \ref{fig:pt} (b), we can conclude that a larger $\gamma$ obtains a better convergence effect in the sense that at the same time instant the state error with $\gamma = 10$ is much less than that using $\gamma = 1$. The argument holds similarly for the case with a larger $h$. In Figure \ref{fig:pt} (d) we fix $\gamma = 10$ and $\bar{T}_s = 4$ while using $h = 20$, comparing with Figure \ref{fig:pt} (c), the convergence effect in this case is even better. 

\subsection{The distributed biased min-consensus protocol under perturbations}
\begin{figure}
	\centering
	\subfigure[DBMC without perturbations using $\eta = 1 + 10^{-8}$ ]{
		\includegraphics[width = 0.9\columnwidth]{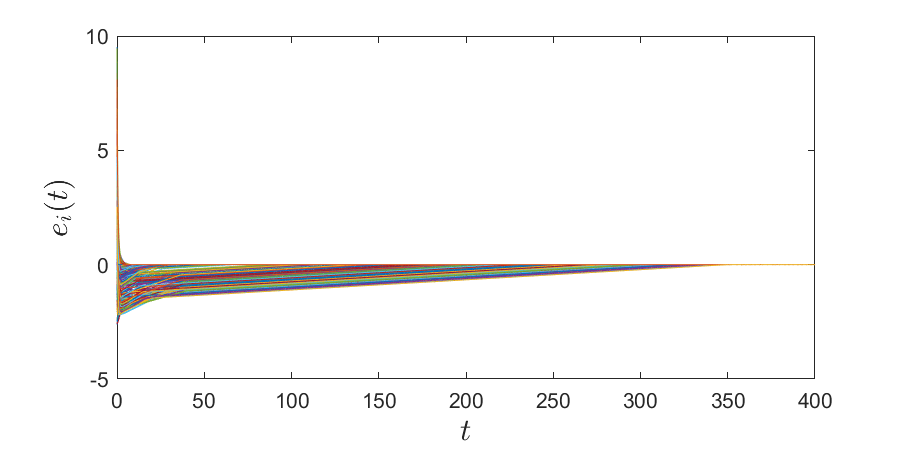}}
	\subfigure[DBMC with $0.9w_{ij} \leq w_{ij}(t) \leq 1.1w_{ij}$ using $\eta = 1 + 10^{-8}$]{  
		\includegraphics[width = 0.9\columnwidth]{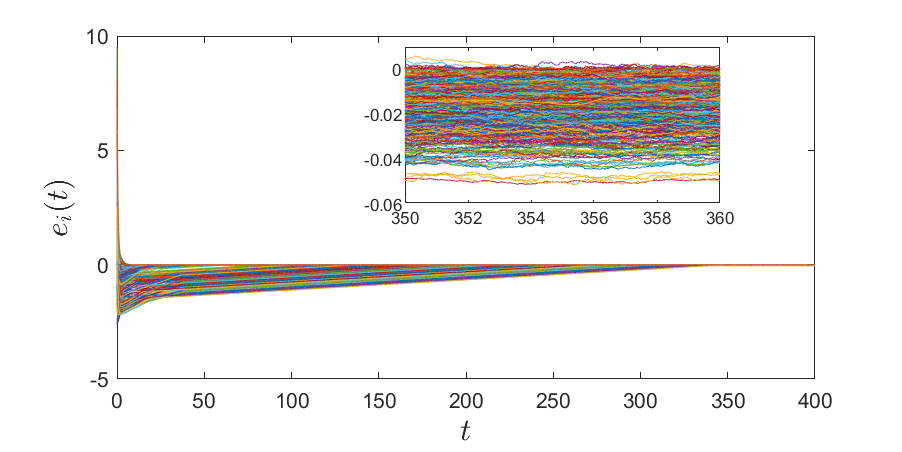}}
	\subfigure[DBMC with $0.8w_{ij} \leq w_{ij}(t) \leq 1.2w_{ij}$ using $\eta = 1 + 10^{-8}$]{  
		\includegraphics[width = 0.9\columnwidth]{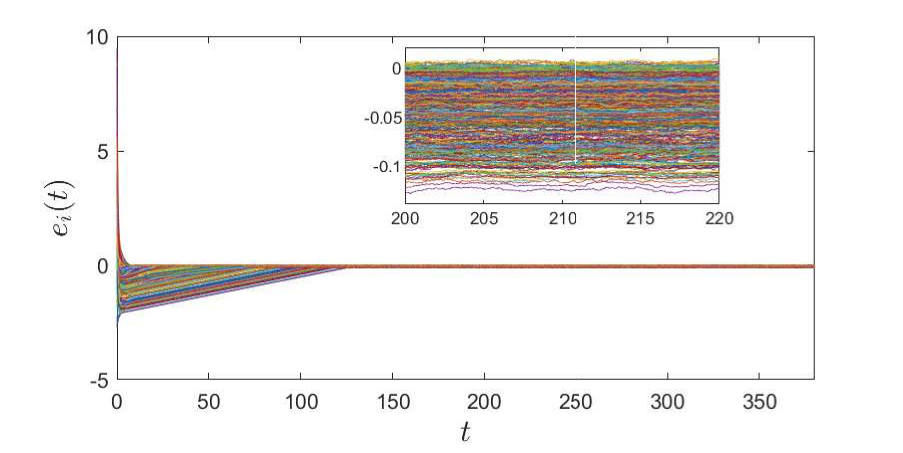}}
	\subfigure[DBMC with $0.9w_{ij} \leq w_{ij}(t) \leq 1.1w_{ij}$ using $\eta = 20$]{  
		\includegraphics[width = 0.9\columnwidth]{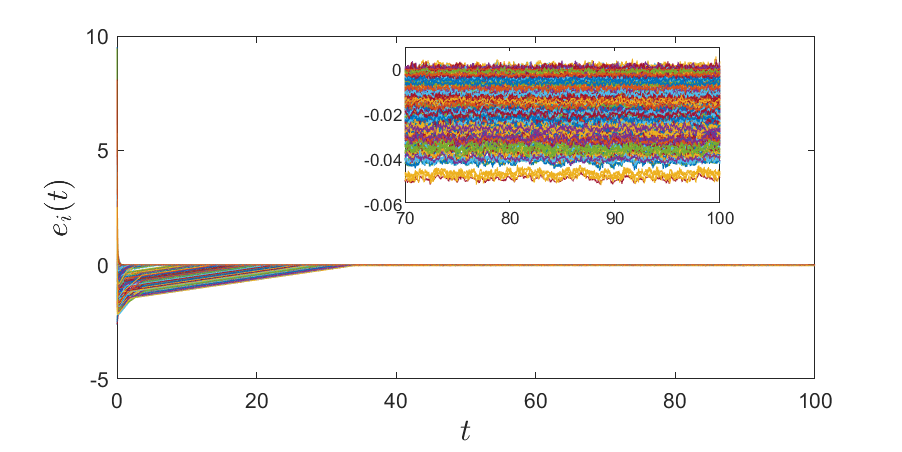}}
	\caption{Results of applying nominal DBMC to randomized graphs with varying $\eta$. With $\eta = 1 + 10^{-8}$, the nominal DBMC achieves exponential stability without perturbations, and expISS with the edge weight $w_{ij}(t)$ ranging from $0.8w_{ij}$ to $1.2w_{ij}$. The nominal DBMC can still achieve expISS with a larger $\eta$, manifesting the conservatism of Lyapunov-based analysis.}
	\label{fig:sg}
\end{figure}

In this part, we look into the behaviors of DBMC, with perturbations including, the change of source and the measurement noise on the edge weights. Unlike the previous two subsection, here we don't have the restriction that all states need to be overestimates, and we assume all initial state are randomly distributed between [0, 10] such that initial states include overestimates and underestimates. We first set $\eta$ in (\ref{eq:practical-protocol}) as $\eta = 1 + 10^{-8}$ (i.e., condition on $\eta$ in Theorem \ref{th:small} is satisfied). Figure \ref{fig:sg} (a) demonstrates the state errors with 0-input, by definition of expISS in Definition \ref{def:expISS}, in this case DBMC converges exponentially fast. Also, states starting with overestimates converge much faster than those with underestimates. Therefore, to achieve a better convergence speed, one can set all initial states in the nominal DBMC to be overestimates. Figure \ref{fig:sg} (b) further provides the trajectories of state errors under perturbations on the edge weights, here we assume each edge weight $w_{ij}(t)$ varies from $0.9w_{ij}$ to $1.1w_{ij}$ continuously. Similarly, all state errors drop below a bound at an exponential speed while those with overestimates decline even faster. It is worth noting that the largest state error is around 0.05 while $u_{ij}(t)$ defined in (\ref{eq:input}) obeys $u_{ij}(t) \leq 0.025$ for all $i \in V$, $j \in \cN(i)$ and all $t \geq 0$, indicating the superior robustness of DBMC. Figure \ref{fig:sg} further consolidates this claim, in this case the edge weight $w_{ij}(t)$ varies from $0.8w_{ij}$ to $1.2w_{ij}$ continuously while $\eta = 1 + 10^{-8}$, the largest state error now is roughly 0.13, implying that it linearly depends on the magnitude of the input. 

We next violate the condition in Theorem \ref{th:small} by increasing $\eta$ from $1 + 10^{-8}$ to $20$ while each edge weight $w_{ij}(t)$ still varies from $0.9w_{ij}$ to $1.1w_{ij}$ continuously. Then it can be seen from Figure \ref{fig:sg} (c) that (\ref{eq:practical-protocol}) is also expISS with all state errors dropping below the bound much more rapidly as the feedback gain increases, manifesting the conservatism of the requirement on $\eta$ in Theorem \ref{th:small}, originated from the conservative nature of Lyapunov-based analysis.  

Finally we consider the scenario where the source changes during the implementation of the nominal DBMC. In this scenario, $\eta = 1 + 10^{-8}$, the edge weight still bears the previous perturbation and the source node changes from node 1 to node 20 when $t = 380$. As can be seen from Figure \ref{fig:sc}, all state errors decline rapidly till dropping below the bound. Then abrupt spikes emerge at $t = 380$ as now the source becomes node 20, leading to the difference between the estimated distance towards the previous source and the true distance from current source becoming the new initial state errors. As DBMC is globally exponentially input-to-state stable, i.e., all state errors will drop below a bound defined by the deviation of edge weight from its nominal value at an exponential rate regardless of the initial states, then the state errors decline below a bound again after the source change.        
\begin{figure}
	\centering
		\includegraphics[width = 1\columnwidth]{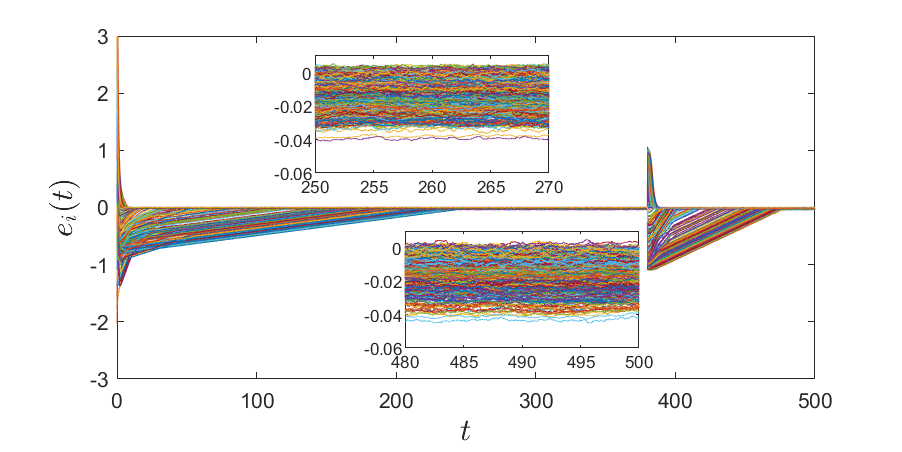}
	\caption{Results of applying nominal DBMC to randomized graphs with the change of source. In this scenario, the source changes from node 1 to node 20 at $t = 380$ while the edge weight $w_{ij}(t)$ ranges from $0.9w_{ij}$ to $1.1w_{ij}$ continuously. (\ref{eq:practical-protocol}) is expISS before and after the source change.}
	\label{fig:sc}
\end{figure}
\section{Conclusions}\label{sec:con}
In this paper, PPT and PT control strategies have been both designed for DBMC. Under the proposed control strategies DBMC can converge to the stationary value exactly or to a certain level adjustable by user-defined parameter, within a presetting finite time. As for the nominal DBMC, we have given the sufficient condition regarding the range of the feedback gain of DBMC to ensure its expISS, i.e., DBMC is globally exponentially stable without perturbations, and with perturbations its state error will drop exponentially fast below a bound defined by the magnitude of the perturbation. An interesting research direction would be designing a two or higher order DBMC, as many systems in engineering are modeled by higher order dynamics, and several works regarding to higher-order min-consensus protocols \cite{zhang2018second,singh2020min} are worthy of reference.   
\bibliographystyle{IEEEtran}
\bibliography{refs}

\begin{thebibliography}{10}
\providecommand{\url}[1]{#1}
\csname url@samestyle\endcsname
\providecommand{\newblock}{\relax}
\providecommand{\bibinfo}[2]{#2}
\providecommand{\BIBentrySTDinterwordspacing}{\spaceskip=0pt\relax}
\providecommand{\BIBentryALTinterwordstretchfactor}{4}
\providecommand{\BIBentryALTinterwordspacing}{\spaceskip=\fontdimen2\font plus
\BIBentryALTinterwordstretchfactor\fontdimen3\font minus
  \fontdimen4\font\relax}
\providecommand{\BIBforeignlanguage}[2]{{%
\expandafter\ifx\csname l@#1\endcsname\relax
\typeout{** WARNING: IEEEtran.bst: No hyphenation pattern has been}%
\typeout{** loaded for the language `#1'. Using the pattern for}%
\typeout{** the default language instead.}%
\else
\language=\csname l@#1\endcsname
\fi
#2}}
\providecommand{\BIBdecl}{\relax}
\BIBdecl

\bibitem{yao2023event}
D.~Yao, H.~Li, and Y.~Shi, ``Event-based average consensus of disturbed mass
  via fully distributed sliding mode control,'' \emph{IEEE Transactions on
  Automatic Control}, 2023.

\bibitem{min-max}
V.~Yadav and M.~Salapaka, ``Distributed protocol for determining when averaging
  consensus is readied,'' in \emph{45th Annual Allerton Conference on
  Communication, Control, and Computing}, 2007, pp. 715--720.

\bibitem{cortesmaxmin}
J.~Cort{\'e}s, ``Distributed algorithms for reaching consensus on general
  functions,'' \emph{Automatica}, vol.~44, no.~3, pp. 726--737, 2008.

\bibitem{Zhang2017}
Y.~{Zhang} and S.~{Li}, ``Distributed biased min-consensus with applications to
  shortest path planning,'' \emph{IEEE Transactions on Automatic Control},
  vol.~62, no.~10, pp. 5429--5436, Oct 2017.

\bibitem{bellman1958routing}
R.~Bellman, ``On a routing problem,'' \emph{Quarterly of applied mathematics},
  vol.~16, no.~1, pp. 87--90, 1958.

\bibitem{dijkstra2022note}
E.~W. Dijkstra, ``A note on two problems in connexion with graphs,'' in
  \emph{Edsger Wybe Dijkstra: His Life, Work, and Legacy}, 2022, pp. 287--290.

\bibitem{russell2016artificial}
S.~J. Russell and P.~Norvig, \emph{Artificial intelligence: a modern
  approach}.\hskip 1em plus 0.5em minus 0.4em\relax Pearson, 2016.

\bibitem{shi2018biased}
X.~Shi, Y.~Xu, and H.~Sun, ``A biased min-consensus-based approach for optimal
  power transaction in multi-energy-router systems,'' \emph{IEEE Transactions
  on Sustainable Energy}, vol.~11, no.~1, pp. 217--228, 2018.

\bibitem{zhang2017perturbing}
Y.~Zhang and S.~Li, ``Perturbing consensus for complexity: A finite-time
  discrete biased min-consensus under time-delay and asynchronism,''
  \emph{Automatica}, vol.~85, pp. 441--447, 2017.

\bibitem{yao2020hierarchical}
P.~Yao, R.~Zhao, and Q.~Zhu, ``A hierarchical architecture using biased
  min-consensus for usv path planning,'' \emph{IEEE Transactions on Vehicular
  Technology}, vol.~69, no.~9, pp. 9518--9527, 2020.

\bibitem{shi2020distributed}
X.~Shi, Y.~Xu, Q.~Guo, H.~Sun, and W.~Gu, ``A distributed ev navigation
  strategy considering the interaction between power system and traffic
  network,'' \emph{IEEE Transactions on Smart Grid}, vol.~11, no.~4, pp.
  3545--3557, 2020.

\bibitem{TAC}
Y.~Mo, S.~Dasgupta, and J.~Beal, ``Robustness of the {Adaptive Bellman-Ford
  Algorithm}: Global stability and ultimate bounds,'' \emph{IEEE Transactions
  on Automatic Control}, pp. 4121--4136, 2019.

\bibitem{mo2022stability}
Y.~Mo, S.~Dasgupta, and J.~Beal, ``Stability and resilience of distributed
  information spreading in aggregate computing,'' \emph{IEEE Transactions on
  Automatic Control}, vol.~68, no.~1, pp. 454--461, 2022.

\bibitem{mo2021lyapunov}
Y.~Mo and L.~Yu, ``A {L}yapunov analysis of the continuous-time adaptive
  bellman--ford algorithm,'' \emph{Systems \& Control Letters}, vol. 157, p.
  105045, 2021.

\bibitem{paulos2019framework}
A.~Paulos, S.~Dasgupta, J.~Beal, Y.~Mo, K.~Hoang, L.~J. Bryan, P.~Pal,
  R.~Schantz, J.~Schewe, R.~Sitaraman \emph{et~al.}, ``A framework for
  self-adaptive dispersal of computing services,'' in \emph{2019 IEEE 4th
  International Workshops on Foundations and Applications of Self* Systems
  (FAS* W)}.\hskip 1em plus 0.5em minus 0.4em\relax IEEE, 2019, pp. 98--103.

\bibitem{song2023prescribed}
Y.~Song, H.~Ye, and F.~L. Lewis, ``Prescribed-time control and its latest
  developments,'' \emph{IEEE Transactions on Systems, Man, and Cybernetics:
  Systems}, 2023.

\bibitem{ning2022fixed}
B.~Ning, Q.-L. Han, Z.~Zuo, L.~Ding, Q.~Lu, and X.~Ge, ``Fixed-time and
  prescribed-time consensus control of multiagent systems and its applications:
  A survey of recent trends and methodologies,'' \emph{IEEE Transactions on
  Industrial Informatics}, vol.~19, no.~2, pp. 1121--1135, 2022.

\bibitem{wang2018leader}
Y.~Wang and Y.~Song, ``Leader-following control of high-order multi-agent
  systems under directed graphs: Pre-specified finite time approach,''
  \emph{Automatica}, vol.~87, pp. 113--120, 2018.

\bibitem{krishnamurthy2020dynamic}
P.~Krishnamurthy, F.~Khorrami, and M.~Krstic, ``A dynamic high-gain design for
  prescribed-time regulation of nonlinear systems,'' \emph{Automatica}, vol.
  115, p. 108860, 2020.

\bibitem{wang2018prescribed}
Y.~Wang, Y.~Song, D.~J. Hill, and M.~Krstic, ``Prescribed-time consensus and
  containment control of networked multiagent systems,'' \emph{IEEE
  transactions on cybernetics}, vol.~49, no.~4, pp. 1138--1147, 2018.

\bibitem{cao2022practical}
Y.~Cao, J.~Cao, and Y.~Song, ``Practical prescribed time tracking control over
  infinite time interval involving mismatched uncertainties and non-vanishing
  disturbances,'' \emph{Automatica}, vol. 136, p. 110050, 2022.

\bibitem{ning2019practical}
B.~Ning, Q.-L. Han, and Z.~Zuo, ``Practical fixed-time consensus for
  integrator-type multi-agent systems: A time base generator approach,''
  \emph{Automatica}, vol. 105, pp. 406--414, 2019.

\bibitem{geiselhart2016relaxed}
R.~Geiselhart and F.~R. Wirth, ``Relaxed iss small-gain theorems for
  discrete-time systems,'' \emph{SIAM Journal on Control and Optimization},
  vol.~54, no.~2, pp. 423--449, 2016.

\bibitem{dashkovskiy2010small}
S.~N. Dashkovskiy, B.~S. R{\"u}ffer, and F.~R. Wirth, ``Small gain theorems for
  large scale systems and construction of iss lyapunov functions,'' \emph{SIAM
  Journal on Control and Optimization}, vol.~48, no.~6, pp. 4089--4118, 2010.

\bibitem{ruffer2010monotone}
B.~S. R{\"u}ffer, ``Monotone inequalities, dynamical systems, and paths in the
  positive orthant of euclidean n-space,'' \emph{Positivity}, vol.~14, pp.
  257--283, 2010.

\bibitem{khalil}
H.~K. Khalil, ``Nonlinear systems,'' \emph{Upper Saddle River}, 2002.

\bibitem{morasso1997computational}
P.~Morasso, V.~Sanguineti, and G.~Spada, ``A computational theory of targeting
  movements based on force fields and topology representing networks,''
  \emph{Neurocomputing}, vol.~15, no. 3-4, pp. 411--434, 1997.

\bibitem{clarke1990optimization}
F.~H. Clarke, \emph{Optimization and {N}onsmooth {A}nalysis}.\hskip 1em plus
  0.5em minus 0.4em\relax SIAM, 1990.

\bibitem{stewart2007multivariable}
J.~Stewart, \emph{Multivariable Calculus: Early Transcendentals}.\hskip 1em
  plus 0.5em minus 0.4em\relax Cengage Learning, 2007.

\bibitem{haddad2008nonlinear}
W.~M. Haddad and V.~Chellaboina, \emph{Nonlinear dynamical systems and control:
  a Lyapunov-based approach}.\hskip 1em plus 0.5em minus 0.4em\relax Princeton
  university press, 2008.

\bibitem{Mironchenko2023}
A.~Mironchenko, \emph{Input-to-State Stability}.\hskip 1em plus 0.5em minus
  0.4em\relax Cham: Springer International Publishing, 2023, pp. 41--115.

\bibitem{berman1994nonnegative}
A.~Berman and R.~J. Plemmons, \emph{Nonnegative matrices in the mathematical
  sciences}.\hskip 1em plus 0.5em minus 0.4em\relax SIAM, 1994.

\bibitem{geiselhart2012numerical}
R.~Geiselhart and F.~Wirth, ``Numerical construction of liss lyapunov functions
  under a small-gain condition,'' \emph{Mathematics of Control, Signals, and
  Systems}, vol.~24, pp. 3--32, 2012.

\bibitem{zhang2018second}
Y.~Zhang and S.~Li, ``Second-order min-consensus on switching topology,''
  \emph{Automatica}, vol.~96, pp. 293--297, 2018.

\bibitem{singh2020min}
B.~Singh, A.~Sen, and S.~R. Sahoo, ``Min-consensus for heterogeneous
  higher-order integrators under switching digraph,'' \emph{IEEE Control
  Systems Letters}, vol.~4, no.~3, pp. 560--565, 2020.

\end{thebibliography}

\appendix
\noindent

{\bf Proof of Theorem \ref{the:unique}:} It follows from Remark \ref{re:TBG} that $0 \leq \eta(t) \leq L$ for all $t \geq 0$ with some $L > 0$. For $i \in V$, define a nonlinear map $f_i: [0,+\infty) \times \mathbb{R}^n \rightarrow \mathbb{R}$ such that $f_i = 0$ for $i \in S$, and for $i \in V\setminus S$, $f_i(t, \tilde{x}) = -\eta(t) (\tilde{x}_i - \min_{j \in \cN(i)} \{\tilde{x}_j + w_{ij} \})$ with $\tilde{x} = [\tilde{x}_1, \cdots, \tilde{x}_n]^\top \in \mathbb{R}^n$ and $i \in \{1,\cdots,n \}$, and the composite map is defined as $f: [0,+\infty) \times \mathbb{R}^n \rightarrow \mathbb{R}^n$. Then for any $\xb = [\xb_1, \cdots, \xb_n]^\top \in \mathbb{R}^n$, there holds
\begin{flalign}
& |f(t,\tilde{x}) - f(t,\xb)|_\infty = |f_i(t,\tilde{x}) - f_i(t,\xb)|  \label{eq:composite} \\
&\leq L|(\tilde{x}_i - \min_{j \in \cN(i)} \{\tilde{x}_j + w_{ij} \}) - (\xb_i - \min_{j \in \cN(i)} \{\xb_j + w_{ij} \})| \nonumber \\
& \leq L|\tilde{x}_i - \xb_i| + L| \min_{j \in \cN(i)} \{\xb_j + w_{ij} \} - \min_{j \in \cN(i)} \{\tilde{x}_j + w_{ij} \} | \nonumber \\
& \leq L|\tilde{x}_i - \xb_i| + L|\xb_{i'} + w_{i,i'} - (\tilde{x}_{j'} + w_{i,j'})| \label{eq:argm} \\
& \leq L|\tilde{x}_i - \xb_i| + L|\xb_{j'} + w_{i,j'} - (\tilde{x}_{j'} + w_{i,j'})| \label{eq:tri} \\
& \leq 2L|\tilde{x} - \xb|_\infty \nonumber
\end{flalign}
where in (\ref{eq:composite}) we assume $i = \arg \max_{j \in \{1,\cdots,n \}}\{|f_j(t,\tilde{x}) - f_j(t,\xb)| \}$, in (\ref{eq:argm}) we assume $i' = \arg \min_{j \in \cN(i)} \{\xb_j + w_{ij} \}$, $j' =\arg \min_{j \in \cN(i)} \{\tilde{x}_j + w_{ij} \}$, and (\ref{eq:tri}) uses the fact that $\xb_{i'} + w_{i,i'} \leq \xb_{j} + w_{i,j}$ for all $j \in \cN(i)$ and the fact that node $j'$ is a neighbor of node $i$. Therefore, the composite map $f(\cdot,\cdot)$ is globally Lipschitz continuous with respect to its second argument for $t \geq 0$ and continuous with respect to its first argument, from Theorem 3.2 in \cite{khalil}, (\ref{eq:protocol}) admits a unique solution.

{\bf Proof of Lemma \ref{le:practical}:} It can be readily verified that the solution to (\ref{eq:TBGy}) over the interval $[0, T_s]$ obeys
\begin{equation}
y(t) = y_0\left(1 - \frac{\varepsilon(t)}{1 + \delta} \right),
\end{equation}
and thus we have $y(T_s) = \frac{\delta}{1 + \delta}y(0)$.

{\bf Proof of Lemma \ref{le:bound}:} 
We prove our claim by induction. We begin with $\ell = 1$, for $i \in \cF_1$, it has a true parent node $j \in S = \cF_0$, then it follows from (\ref{eq:protocol}) and (\ref{eq:error}) that for all $t \in [0, T_s]$
\begin{flalign}
	&\ed_{i}(t) = \xd_i(t) =  -\eta(t)\left(x_i(t) - \min_{k \in \cN(i)}\{x_k(t)+w_{ik}\}\right) \nonumber \\
	& = -\eta(t)\left(x_i(t) - (x_j + w_{ij}) \right)\label{eq:bov} \\
	&= -\eta(t)e_i(t) \label{eq:ust}
\end{flalign}
where (\ref{eq:bov}) uses Lemma \ref{le:allover} such that $x_i(t) \geq x_i$ for all $i \in V$ and all $t \geq 0$, and the fact that $j \in S$ is the true parent node of $i$, i.e., $x_j(t) + w_{ij} = x_j + w_{ij} = w_{ij} \leq x_k + w_{ik} \leq x_k(t) + w_{ik}$ for all $k \in \cN(i)$, and (\ref{eq:ust}) uses (\ref{eq:bellman}). Further, with Lemma \ref{le:practical}, for $t \in [0, T_s]$ there holds
\begin{equation}
e_i(t) = e_i(0)\left(1 - \frac{\varepsilon(t)}{1 + \delta} \right)\leq e_{\max}(0)\left(1 - \frac{\varepsilon(t)}{1 + \delta} \right),
\end{equation} 
and thus $e_i(T_s) \leq \frac{\delta}{1 + \delta}e_{\max}(0)$. 

We next prove that $e_i(t) \leq e_i(T_s)$ for all $t \geq T_s$. Note that $\ed_{i}(t) = -\eta(t)e_i(t)$ still holds for $t \geq T_s$. Applying Lemma \ref{le:practical} on $\ed_{i}(t)$ over the interval $[T_s, 2T_s]$, we obtain 
\begin{flalign}
e_i(t) = e_i(T_s)\left(1 - \frac{\varepsilon(t - T_s)}{1 + \delta} \right) \leq e_i(T_s), 
\end{flalign}
repeating the above procedure, there holds $e_i(t) \leq e_i(T_s) = \frac{\delta}{1 + \delta}e_{\max}(0)$ for all $t \geq T_s$, and (\ref{eq:instant}) holds with $\ell = 1$.

Now suppose (\ref{eq:instant}) holds with some $\ell \in \{ 1,\cdots, \cD(G)-2   \}$. For $i \in \cF_{\ell + 1}$, it follows from (\ref{eq:protocol}) and (\ref{eq:error}) that for all $t \in [\ell T_S, (\ell + 1)T_s]$, we have
\begin{flalign}
&\ed_{i}(t) = \xd_i(t)  = -\eta(t)\left(x_i(t) - \min_{k \in \cN(i)}\{x_k(t)+w_{ik}\}\right)\nonumber \\
& \leq  -\eta(t)(x_i(t) - (x_j(t)+w_{ij})) \label{eq:asj} \\
& = -\eta(t)(x_i + e_i(t) - (x_j + e_j(t)+w_{ij})) \nonumber \\
& \leq -\eta(t)(e_i(t) - \frac{\ell\delta}{1+\delta}e_{\max}(0)) \label{eq:ell}
\end{flalign}
where in (\ref{eq:asj}) we assume $j$ is the true parent node of $i$, and (\ref{eq:ell}) uses (\ref{eq:bellman}), the fact that $j \in \cF_\ell$, as well as our induction hypothesis that $e_j(t) \leq \frac{\ell\delta}{1+\delta}e_{\max}(0)$ for all $t \geq \ell T_s$.

As (\ref{eq:ell}) implies $(e_i(t) - \ell\frac{\delta}{1+\delta}e_{\max}(0))' = -\eta(t)(e_i(t) - \ell\frac{\delta}{1+\delta}e_{\max}(0)) $, by applying Lemma \ref{le:practical} and comparison principle \cite{khalil}, we obtain for $t \in [\ell T_S, (\ell + 1)T_s]$
\begin{flalign}
&e_i(t) - \ell\frac{\delta}{1+\delta}e_{\max}(0) \nonumber \\
&\leq\left(e_i(\ell T_s) - \ell\frac{\delta}{1+\delta}e_{\max}(0)\right)\left(1 - \frac{\varepsilon(t-\ell T_s)}{1 + \delta} \right), \label{eq:ine}
\end{flalign}
which further leads to 
\begin{flalign}
&e_i((\ell+1)T_s)  \nonumber \\ &\leq\left(e_{\max}(0) \!-\! \frac{\ell\delta}{1+\delta}e_{\max}(0)\right)\frac{\delta}{1 + \delta} + \frac{\ell\delta}{1+\delta}e_{\max}(0) \label{eq:vp} \\
& \leq \frac{(\ell+1)\delta}{1+\delta}e_{\max}(0) \nonumber 
\end{flalign}
where (\ref{eq:vp}) uses Lemma \ref{le:greatest} and $\varepsilon(T_s) = 1$. To complete our proof, we need to show that $e_i(t) \leq e_i((\ell+1)T_s)$ for all $t \geq (\ell+1)T_s$. Again, note that (\ref{eq:ell}) still holds over $[(\ell + 1)T_s, (\ell + 2)T_s]$, it follows from (\ref{eq:ine}) that for $t \in [(\ell + 1)T_s, (\ell + 2)T_s]$
\begin{flalign}
&e_i(t) - \ell\frac{\delta}{1+\delta}e_{\max}(0) \nonumber \\
&\leq\left(e_i((\ell+1) T_s) - \ell\frac{\delta}{1+\delta}e_{\max}(0)\right)\left(1 - \frac{\varepsilon(t-(\ell+1) T_s)}{1 + \delta} \right) \nonumber \\
& \leq e_i((\ell+1) T_s) - \ell\frac{\delta}{1+\delta}e_{\max}(0) \label{eq:varl}
\end{flalign}
where (\ref{eq:varl}) uses the fact that $\varepsilon(t) \in [0, 1]$ for $t \in [0,T_s]$. Continuing the above procedure, we can conclude that $e_i(t) \leq \frac{(\ell+1)\delta}{1+\delta}e_{\max}(0)$ for $t \geq (\ell + 1)T_s$, establishing our claim.

{\bf Proof of Lemma \ref{le:limi}:}
Multiplying $\rho^\alpha(t)$ on both sides of (\ref{eq:y}), we obtain 
\begin{equation}
\rho^\alpha(t)\dot{y}(t) = -\gamma\rho^\alpha(t)y(t) - \alpha\rho(t)^{\alpha - 1}\dot{\rho}(t)y(t)
\end{equation}
which further gives 
\begin{equation}\label{eq:cy}
\frac{d(\rho^\alpha(t)y(t))}{dt} =  \rho^\alpha(t)\dot{y}(t) + \alpha\rho(t)^{\alpha - 1}\dot{\rho}(t)y(t) = -\gamma\rho^\alpha(t)y(t)
\end{equation}
Applying comparison principle on (\ref{eq:cy}) gives 
\begin{equation}\label{eq:rho0}
\rho^\alpha(t)y(t) = e^{-\gamma t}\rho^\alpha(0)y_0 = e^{-\gamma t}y_0, \forall t \in [0, \bar{T}_s),
\end{equation}
where the last equality uses $\rho(0) = 1$.

{\bf Proof of Lemma \ref{le:leftlimit}:} Consider a sequence of nodes $i_0, i_1, \cdots, i_{T}$ such that $i_\ell \in \cF_\ell$ with $\ell \in \{0,\cdots,T\}$ and $i_\ell$ is the true parent node of $i_{\ell + 1}$ with $\ell \in \{0,\cdots,T-1\}$. Every node in the graph $G$ is in such a sequence and $T \leq \cD(G) - 1$. We begin with node $i_2$. It follows from (\ref{eq:prespecified}) and (\ref{eq:error}) that for $t \in [0, \bar{T}_s)$
\begin{flalign}
& \ed_{i_2}(t) = \xd_{i_2}(t) \nonumber \\
& = -\bar{\eta}(t)\left(x_{i_2}(t) - \min_{j \in \cN(i_2)}\{x_j(t)+w_{i_2j}\}\right) \nonumber\\
& = -\bar{\eta}(t)\left(x_{i_2}(t) - (x_{i_1} + w_{i_1i_2}) \right) \label{eq:i2} \\
&= -\bar{\eta}(t)e_{i_2}(t) \label{eq:red}
\end{flalign}
where (\ref{eq:i2}) uses Lemma \ref{le:alloverpt} such that $x_i(t) \geq x_i$ for all $i \in V$ and all $t \geq 0$, and the fact that $i_1 \in S$ is the true parent node of $i_2$, i.e., $x_{i_1}(t) + w_{i_1i_2} = x_{i_1} + w_{i_1i_2} = w_{i_1i_2} \leq x_j + w_{i_2j} \leq x_j(t) + w_{i_2j}$ for all $j \in \cN(i_2)$, and (\ref{eq:red}) uses (\ref{eq:bellman}). 

From Lemma \ref{le:limi}, (\ref{eq:red}) gives
\begin{equation}\label{eq:form}
e_{i_2}(t) = \rho^{-2}(t)e^{-\gamma t}e_{i_2}(0),
\end{equation}
leading to $\lim_{t \rightarrow \bar{T}_s^-} e_{i_2}(t) = 0,$ and thus $\lim_{t \rightarrow \bar{T}_s^-} x_{i_2}(t) = x_{i_2}$.

For node $i_3$, we prove our claim using contradiction. Suppose $\lim_{t \rightarrow \bar{T}_s^-} e_{i_3}(t) \neq 0$. As $e_i(t) \geq 0$ for all $i \in V$ and $t \in [0, \bar{T}_s)$ given in Lemma \ref{le:alloverpt},
this assumption is equivalent to $\exists \epsilon_1 > 0$, such that $\forall \delta_1 > 0$, there exists a $t_1 \in [\bar{T}_s - \delta_1, \bar{T}_s)$ such that $e_{i_3}(t) \geq \epsilon_1$. As $\lim_{t \rightarrow \bar{T}_s^-} e_{i_2}(t) = 0$ implies that for some $\epsilon_2 = \frac{1}{3}\epsilon_1$, $\exists \delta_2 > 0$ such that $e_{i_2}(t) < \epsilon_2$ for all $t \in [\bar{T}_s - \delta_2, \bar{T}_s)$, it follows from (\ref{eq:prespecified}) and (\ref{eq:error}) that for $t \in [\bar{T}_s - \delta_2, \bar{T}_s)$
\begin{flalign}
& \ed_{i_3}(t) = -\bar{\eta}(t)\left(x_{i_3}(t) - \min_{j \in \cN(i_3)}\{x_j(t)+w_{i_3j}\}\right) \nonumber\\
&\leq -\bar{\eta}(t)(x_{i_3}(t) - (x_{i_2}(t)+w_{i_2i_3})) \nonumber \\
&\leq -\bar{\eta}(t)(x_{i_3}(t) - (x_{i_2}+ \epsilon_2 + w_{i_2i_3})) \nonumber \\
& = -\bar{\eta}(t) (e_{i_3}(t) - \epsilon_2  )\label{eq:long}
\end{flalign}
where (\ref{eq:long}) uses (\ref{eq:bellman}). Using comparison principle and Lemma \ref{le:limi}, there holds
\begin{equation}\label{eq:ft}
e_{i_3}(t) \leq \underbrace{\rho^{-2}(t)e^{-\gamma (t - (\bar{T}_s - \delta_2))}e_{i_3}(\bar{T}_s - \delta_2)}_{\Delta(t)} + \epsilon_2.
\end{equation}
As the left-hand limit at $\bar{T}_s$ of $\Delta(t)$ of (\ref{eq:ft}) is 0  and $\Delta(t) \geq 0$ for $t\in [\bar{T}_s - \delta_2, \bar{T}_s)$, we have for $\epsilon_2 > 0$, $\exists \delta_3 > 0 (\delta_3 < \delta_2)$ such that $\Delta(t) < \epsilon_2$ for all $t \in [\bar{T}_s - \delta_3, \bar{T}_s)$, leading to $e_{i_3}(t) \leq \epsilon_2 = \frac{2}{3}\epsilon_1$ during such an interval, establishing the contradiction, and thus $\lim_{t \rightarrow \bar{T}_s^-} x_{i_3}(t) = x_{i_3}$.

Repeating the above steps on the remaining nodes, it can be verified that $\lim_{t \rightarrow \bar{T}_s^-} x_{i_\ell}(t) = x_{i_\ell}$ for all $\ell \in \{1,\cdots,T \}$, completing our proof.
\begin{figure}[h]
	\centering
	\includegraphics[width=1\columnwidth]{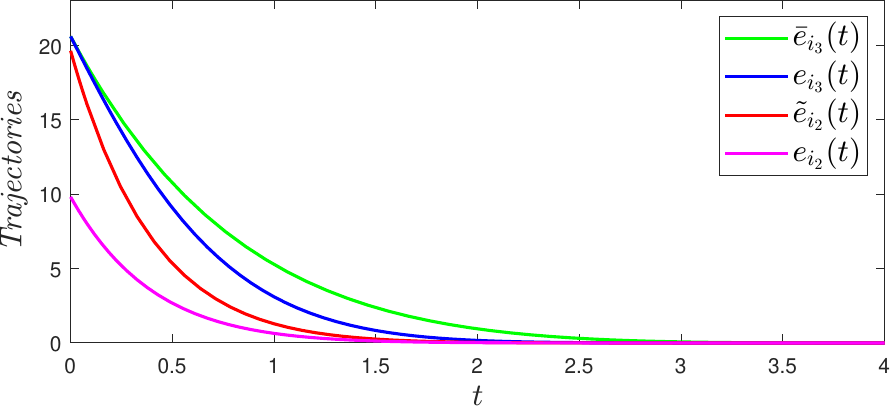}
	\caption{Illustration of the proof of Lemma \ref{le:continuity} for the case 1): $e_{i_3}(0) \geq ke_{i_2}(0)$, via implementing DBMC with the PT control strategy on an undirected graph. In this scenario, $k = 2$, $e_{i_3}(0) = 20.6059 > 2*e_{i_2}(0) = 19.6381$, and $t_1$ here obeys $t_1 = \bar{T}_s = 4$.}
	\label{fig:traj}
\end{figure}

{\bf Proof of Lemma \ref{le:continuity}:}
We use the sequence of nodes $i_0, i_1, \cdots, i_{T}$ as considered in the proof of Lemma \ref{le:leftlimit}. We begin with node $i_2$, it follows from (\ref{eq:red}) and (\ref{eq:form}) in the proof of Lemma \ref{le:leftlimit} that for $t \in [0, \bar{T}_s)$
\begin{flalign}
&\ed_{i_2}(t)  =  -\bar{\eta}(t)\rho^{-2}(t)e^{-\gamma t}e_{i_2}(0) \nonumber \\
& =  - \left( \gamma\frac{(\bar{T}_s - t)^{2 + 2h}}{\bar{T}_s^{2 + 2h}} + 2\frac{((1 + h)\bar{T}_s - t)^{1 + 2h}}{\bar{T}_s^{2 + 2h}}  \right)e^{-\gamma t}e_{i_2}(0),\nonumber
\end{flalign}
implying $\lim_{t \rightarrow \bar{T}_s^-} \xd_{i_2}(t) = \ed_{i_2}(t) = 0.$

We prove (\ref{eq:upper}) by using node $i_3$ as an example, the proof of other nodes is in a similar fashion and thus omitted.

Consider node $i_3$, we consider two cases: 1) $e_{i_3}(0) \geq ke_{i_2}(0)$ for some $k \geq 2$; 2) $e_{i_3}(0) \leq ke_{i_2}(0)$.  In the first case, we define the following three variables over the interval $[0, \bar{T}_s)$:
\begin{eqnarray}\label{eq:trva}
	&&\dot{\bar{e}}_{i_3}(t) = -\bar{\eta}(t)(1 - \frac{1}{k})\bar{e}_{i_3}(t),\nonumber \\
	&&\dot{\tilde{e}}_{i_3}(t) = -\bar{\eta}(t)\tilde{e}_{i_3}(t),\nonumber \\
	&& \dot{\tilde{e}}_{i_2}(t) = -\bar{\eta}(t)\tilde{e}_{i_2}(t)
\end{eqnarray}
with $\bar{e}_{i_3}(0) = \tilde{e}_{i_3}(0) = e_{i_3}(0)$ and $\tilde{e}_{i_2}(t) = ke_{i_2}(0)$.  Using comparison principle and Lemma \ref{le:limi} yields for all $t \in [0, \bar{T}_s)$
\begin{equation}\label{eq:treq}
ke_{i_2}(t) \!= \!	\tilde{e}_{i_2}(t)\! \leq\! \tilde{e}_{i_3}(t) \!\leq\! \bar{e}_{i_3}(t) \!=\! \frac{1}{\rho^{2(1 - \frac{1}{k}) }(t) }e^{-\gamma(1 - \frac{1}{k})t}e_{i_3}(0)
\end{equation}
We will prove that in case 1) $e_{i_3}(t) \leq \bar{e}_{i_3}(t)$ for all $t \in [0, \bar{T}_s)$.

By the continuity of $e_{i_3}(t)$ and $e_{i_2}(t)$ over $[0, \bar{T}_s)$, there exists a $t_1 \leq \bar{T}_s$ such that $e_{i_3}(t) \geq ke_{i_2}(t)$ for $[0, t_1]$. Then it follows from (\ref{eq:prespecified}) that over $[0, t_1)$
\begin{flalign}
	& \ed_{i_3}(t) = 	-\bar{\eta}(t)\left(x_{i_3}(t) - \min_{j \in \cN(i_3)}\{x_j(t)+w_{i_3j}\}\right) \nonumber \\
    &\leq -\bar{\eta}(t) (  x_{i_3}(t) - ( x_{i_2}(t) + w_{i_2i_3} ) ) \nonumber \\
    & = -\bar{\eta}(t)(e_{i_3}(t) - e_{i_2}(t))  \label{eq:ge} \\
    & \leq -\bar{\eta}(t)(1 - \frac{1}{k})e_{i_3}(t)\label{eq:uk}
\end{flalign}
(\ref{eq:ge}) uses that $i_2$ is a true parent node of $i_3$,  and (\ref{eq:uk}) uses $e_{i_3}(t) \geq ke_{i_2}(t)$, leading to $e_{i_3}(t) \leq \bar{e}_{i_3}(t)$ for $t \in [0, t_1)$.  Now suppose there exists a $\bar{T}_s \geq t_2 > t_1$ such that $e_{i_3}(t) \leq ke_{i_2}(t) = \tilde{e}_{i_2}(t)$, there holds $e_{i_3}(t) \leq \bar{e}_{i_3}(t)$ by (\ref{eq:treq}). 
Next suppose there exists a $\bar{T}_s \geq t_3 > t_2 $ such that $e_{i_3}(t) \geq ke_{i_2}(t)$ again for $t \in [t_2, t_3]$, it follows from (\ref{eq:ge}), (\ref{eq:uk}) and Lemma \ref{le:limi} that for all $t \in [t_2, t_3]$
\begin{flalign}
&	e_{i_3}(t) \leq \frac{1}{\rho^{2(1 - \frac{1}{k}) }(t) }e^{-\gamma(1 - \frac{1}{k})(t - t_2)}\rho^{2(1 - \frac{1}{k}) }(t_2)\bar{e}_{i_3}(t_2) \nonumber \\
&   \leq \frac{1}{\rho^{2(1 - \frac{1}{k}) }(t) }e^{-\gamma(1 - \frac{1}{k})(t - t_2)} \rho^{2(1 - \frac{1}{k}) }(t_2)\cdot \nonumber \\
& ~~~ \frac{1}{\rho^{2(1 - \frac{1}{k}) }(t_2) }e^{-\gamma(1 - \frac{1}{k})t_2}e_{i_3}(0)  \label{eq:usdef} \\
&  = \bar{e}_{i_3}(t) \nonumber 
\end{flalign} 
where (\ref{eq:usdef}) uses (\ref{eq:treq}). Along this line, it can be established that $e_{i_3}(t) \leq \bar{e}_{i_3}(t)$ for $t \in [0, \bar{T}_s)$.

For the latte case, we further define
\begin{equation}
	\ed'_{i_3}(t) = -\bar{\eta}(t)(1 - \frac{1}{k})e'_{i_3}(t)
\end{equation} 
with $e'_{i_3}(0) = ke_{i_2}(0)$. Then there holds 
\begin{equation}\label{eq:eprime}
	\frac{1}{\rho^{2(1 - \frac{1}{k}) }(t) }e^{-\gamma(1 - \frac{1}{k})t}ke_{i_2}(0) = e'_{i_3}(t) \geq  \tilde{e}_{i_2}(t) = ke_{i_2}(t)
\end{equation}
for all $t \in [0, \bar{T}_s)$, with $\tilde{e}_{i_2}(t)$ in (\ref{eq:trva}). We will prove that in case 2) $e_{i_3}(t) \leq e'_{i_3}(t)$ for all $t \in [0, \bar{T}_s)$.

By the continuity of $e_{i_3}(t)$ and $e_{i_2}(t)$ over $[0, \bar{T}_s)$, there exists a $t_1 \leq \bar{T}_s$ such that $e_{i_3}(t) \leq ke_{i_2}(t) \leq e'_{i_3}(t)$ for $[0, t_1]$. Now suppose there exists a $\bar{T}_s \geq t_2 > t_1$ such that $e_{i_3}(t) \geq ke_{i_2}(t)$ for $t \in [t_1, t_2]$. As $e_{i_3}(t_1) = ke_{i_2}(t_1) \leq e'_{i_3}(t_1)$ It follows from (\ref{eq:ge}), (\ref{eq:uk}) and Lemma \ref{le:limi} that for all $t \in [t_1, t_2]$
\begin{flalign}
	&	e_{i_3}(t) \leq \frac{1}{\rho^{2(1 - \frac{1}{k}) }(t) }e^{-\gamma(1 - \frac{1}{k})(t - t_1)}\rho^{2(1 - \frac{1}{k}) }(t_1)e'_{i_3}(t_1) \nonumber \\
	& \leq \frac{1}{\rho^{2(1 - \frac{1}{k}) }(t) }e^{-\gamma(1 - \frac{1}{k})(t - t_1)}\rho^{2(1 - \frac{1}{k}) }(t_1) \cdot \nonumber \\
	& ~~~\frac{1}{\rho^{2(1 - \frac{1}{k}) }(t_1) }e^{-\gamma(1 - \frac{1}{k})t_1}ke_{i_2}(0) \nonumber \\
	& = e'_{i_3}(t)
\end{flalign}
Along this line, we can conclude that $e_{i_3}(t) \leq e'_{i_3}(t)$ for $t \in [0, \bar{T}_s)$. Combining (\ref{eq:treq}) and (\ref{eq:eprime}), our claim follows.

{\bf Proof of Lemma \ref{le:squeeze}:} It follows from (\ref{eq:prespecified}) that for $i \in \cF_\ell$ with  $\ell \in \{2,\cdots, \cD(G) - 1 \}$
\begin{flalign}
&	\xd_i(t) = \ed_i(t) = -\eta(t)\left(x_i(t) - \min_{j \in \cN(i)}\{x_j(t)+w_{ij}\}\right) \nonumber \\
& \geq -\bar{\eta}(t)(x_i(t) - x_i) = -\bar{\eta}(t)e_i(t) \label{eq:uall}
\end{flalign}
where (\ref{eq:uall}) uses Lemma \ref{le:alloverpt} such that $\min_{j \in \cN(i)}\{x_j(t)+w_{ij}\} \geq \min_{j \in \cN(i)}\{x_j+w_{ij}\} = x_i$. Putting (\ref{eq:upper}) into (\ref{eq:uall}) yields
\begin{flalign}
&	\ed_i(t) \geq -\bar{\eta}(t)\frac{1}{\rho^{2(1 - \frac{1}{k}) }(t) }e^{-\gamma(1 - \frac{1}{k})t}ke_{\max}(0) \nonumber \\
& =	-\left(\gamma\left(\frac{\bar{T}_s - t}{\bar{T}_s} \right)^{\Gamma(h,k)} +\frac{(2 + 2h)(\bar{T}_s - t)^{\Gamma(h,k)-1}}{\bar{T}_s^{\Gamma(h,k)}} \right)\cdot  \nonumber\\
& ~~~~e^{-\gamma(1 - \frac{1}{k})t}ke_{\max}(0) \label{eq:com}
\end{flalign}
where in (\ref{eq:com}) $\Gamma(h,k) = (2+2h)(1 - 1/k)$. As $k \geq 2$ and $h > 1 (h \in \mathbb{Z}_+)$, there holds $\lim_{t \rightarrow \bar{T}_s^-}\ed_i(t) \geq 0$.

To derive the lower bound of the left-hand limit of $\ed_i(t)$ at $\bar{T}_s$, note that (\ref{eq:uall}) implies $e_i(t)\geq \frac{1}{\rho^2(t)}e^{-\lambda t}e_i(0)$ via Lemma \ref{le:limi}, together with (\ref{eq:prespecified}), there holds
\begin{flalign}
	& \ed_{i}(t) = 	-\bar{\eta}(t)\left(x_{i}(t) - \min_{j \in \cN(i)}\{x_j(t)+w_{ij}\}\right) \nonumber \\
	&\leq -\bar{\eta}(t) (  x_{i}(t) - ( x_{j}(t) + w_{ij} ) ) \label{eq:ag} \\
	& = -\bar{\eta}(t)(e_{i}(t) - e_{j}(t)) \label{eq:utt}  \\
	& \leq -\bar{\eta}(t)\left(\frac{1}{\rho^2(t)}e^{-\lambda t}e_i(0) +  \frac{1}{\rho^{2(1 - \frac{1}{k}) }(t) }e^{-\gamma(1 - \frac{1}{k})t}ke_{\max}(0) \right) \label{eq:ute}
\end{flalign}
where in (\ref{eq:ag}) we assume $j$ is true parent node of $i$, and (\ref{eq:utt}) uses (\ref{eq:bellman}), and (\ref{eq:ute}) uses Lemma \ref{le:continuity}. Then it follows from (\ref{eq:ebar}) and (\ref{eq:com}) that $\lim_{t \rightarrow \bar{T}_s^-}\ed_i(t) \leq 0$. Therefore, our result holds by Squeeze theorem \cite{stewart2007multivariable}.

{\bf Proof of the claim in Remark \ref{re:ltz}:} It follows from (\ref{eq:abs-protocol}) that for $i \in V\setminus S$
\begin{flalign}
	&\xd_i(t) = -\eta\left(x_i(t) - \min_{j \in \cN(i)}\{|x_j(t)|+w_{ij}(t)\}\right) \nonumber \\
	& \geq -\eta(x_i(t) - w_{\min} ) \label{eq:abs}
\end{flalign}
where (\ref{eq:abs}) uses (\ref{eq:boundnoise}). Applying comparison principle on (\ref{eq:abs}), we obtain 
\begin{equation}
	x_i(t) \geq e^{-\eta t}(x_i(0) - w_{\min}) + w_{\min},
\end{equation}
leading to $x_i(t) \geq 0$ for $t \geq 0$ if $x_i(0) \geq 0$, and $x_i(t) \geq 0$ for $t \geq \frac{1}{\eta}\mathrm{In}\frac{w_{\min} - x_i(0)}{w_{\min}}$ otherwise.

{\bf Proof of Lemma \ref{le:trajectory}:} The proof for (\ref{eq:fors}) is trivial. Consider $i \notin S$ and $e_i(t) > 0$, let $j$ be the true parent node of node $i$, it follows from (\ref{eq:per-protocol}) that
\begin{flalign}
&	\Vd_i(x_i(t)) = \xd_i(t) \nonumber \\
& = -\eta\left(x_i(t) - \min_{j \in \cN(i)}\{x_j(t)+w_{ij}(t)\}\right) \nonumber \\
& \leq -\eta(x_i(t) - (x_j(t)+w_{ij}+u_{ij}(t))) \label{eq:inp} \\
& = -\eta(x_j+w_{ij}+ e_i(t) - (x_j + e_j(t)+w_{ij}+u_{ij}(t))) \label{eq:utr} \\
& = -\eta e_i(t) +\eta e_j(t) + \eta u_{ij}(t) \nonumber \\
& \leq -\eta V_i(e_i(t)) + \eta V_j(e_j(t)) + \eta || u ||_\infty \label{eq:c1}
\end{flalign}
where (\ref{eq:inp}) uses (\ref{eq:per-protocol}) and (\ref{eq:input}), (\ref{eq:utr}) uses (\ref{eq:bellman}), and (\ref{eq:c1}) uses the fact that $e_i(t) > 0$ such that $V_i(e_i(t)) = e_i(t)$.

Suppose $e_i(t) < 0$, let $j$ be the current parent node of node $i$, it follows from (\ref{eq:per-protocol}) that 
\begin{flalign}
	&	\Vd_i(x_i(t)) = -\xd_i(t) \nonumber \\
	& = \eta\left(x_i(t) - \min_{j \in \cN(i)}\{x_j(t)+w_{ij}(t)\}\right) \nonumber \\
	& = \eta(x_i + e_i(t) - (x_j + e_j(t)+w_{ij}+u_{ij}(t))) \nonumber \\
	& \leq \eta\big(\zeta(x_j + w_{ij}) + e_i(t) - \zeta(x_j + e_j(t)+w_{ij}) - u_{ij}(t)   \big) \label{eq:hj} \\
	& = -\eta(-e_i(t)) -\eta\zeta e_j(t) -\eta u_{ij}(t) \nonumber \\
	& = -\eta V_i(e_i(t)) + \eta\zeta V_j(e_j(t)) +\eta|| u ||_\infty \label{eq:c2}
\end{flalign}
where (\ref{eq:hj}) uses (\ref{eq:zeta1}) and the fact that $x_j(t) = x_j + e_j(t) \geq 0, w_{ij} \geq 0$, and (\ref{eq:c2}) uses that $V_i(e_i(t)) = -e_i(t)$ resulted from $e_i(t) < 0$. Combining (\ref{eq:c1}) and (\ref{eq:c2}), (\ref{eq:forns}) holds.

{\bf Proof of Lemma \ref{le:D}:} From Theorem 6.4 in \cite{ruffer2010monotone}, we need to show that all cycles composed of entries in $\Gamma$ defined in (\ref{eq:gain}), each of which multiplied by $d$, the factor of the diagonal operator $D$, are contractions, i.e., 
\begin{equation}\label{eq:contraction}
	d\bar{\lambda}_{i_0i_1} \times d\bar{\lambda}_{i_1i_2} \times \cdots \times d\bar{\lambda}_{i_ri_1} < 1
\end{equation}
where $r = 2,\cdots, n$, $i_j \in \{1,\cdots,n \}$ for $j \in \{1,\cdots, r \}$ and $i_j \neq i_{j'}$ if $j \neq j'.$ If $1 < \eta < 1/\zeta$, there exists a $d_1 > 1$ such that $d_1 \eta \zeta < 1$, and thus it follows from (\ref{eq:gain}) that $d_1\bar{\lambda}_{ij} > 1$ only happens when $i \notin S$ and $j$ is a true parent node of $i$, as in this case $\bar{\lambda}_{ij} = \eta > 1$. Then the only cycle induced by such elements is generated by the following sequence of nodes: $i_1, i_2, \cdots, i_r$, such that $i_\ell$ is the true parent of node $i_{\ell - 1}$ for $\ell \in \{2, \cdots, r\}$. It follows from Definition \ref{def:true} that $r \leq \cD(G)$ and $i_r \in S$, together with (\ref{eq:gain}), we obtain
\begin{flalign}
	&d\bar{\lambda}_{i_1i_2} \times d\bar{\lambda}_{i_2i_3} \times \cdots \times d\bar{\lambda}_{i_ri_1} \leq d^{r} \cdot \eta^{r - 1} \cdot \frac{0.5}{\eta^{\cD(G) - 1}} \nonumber \\
	& \leq 0.5 \cdot d^{\cD(G)}  \label{eq:ed}
\end{flalign}
and thus there must exist a $d_2 > 1$ satisfying $0.5 \cdot d_2^{\cD(G)}  < 1$. Let $d = \min\{d_1,d_2  \}$ be the diagonal operator factor of $D$, our claim follows.
\end{document}